\documentclass[draftcls, onecolumn]{IEEEtran}

\usepackage{ifpdf}

\ifCLASSINFOpdf
\usepackage[pdftex]{graphicx}
\else
\usepackage{graphicx}
\fi

\usepackage{amsmath,amssymb}
\usepackage[ruled, linesnumbered]{algorithm2e}

\usepackage{subfigure}

\usepackage{tikz}

\usetikzlibrary{arrows,shapes,backgrounds}

\tikzstyle{dot}=[circle,draw=gray!80,fill=gray!20,thick,inner sep=0pt,minimum size=12pt] 

\ifpdf
\usepackage[pdftex]{hyperref}
    \hypersetup{
    colorlinks=true,%
    citecolor=black,%
    filecolor=black,%
    linkcolor=black,%
    urlcolor=black
  }  
\else

\fi

\newtheorem{lemma}{Lemma}
\newtheorem{definition}{Definition}

\newtheorem{theorem}[lemma]{Theorem}

\DeclareMathOperator{\head}{\text{head}}
\DeclareMathOperator{\tail}{\text{tail}}
\DeclareMathOperator{\In}{\text{I}}
\DeclareMathOperator{\Out}{\text{O}}

\newcommand{\etal}{\textit{et al.}}

\newcommand{\ffield}{\mathbb{F}}
\newcommand{\fsize}{q}
\newcommand{\mwd}{{\rm MWD}}
\newcommand{\rank}{\mathrm{rank}}
\newcommand{\graph}{\mathcal{G}}
\newcommand{\edgeset}{\mathcal{E}}
\newcommand{\nodeset}{\mathcal{V}}
\newcommand{\sinkset}{\mathcal{T}}
\newcommand{\errset}{\ffield^{|\edgeset|}}

\newcommand{\sinknode}{t}
\newcommand{\sinkt}{t}

\newcommand{\msgdim}{\omega}

\newcommand{\msgset}{\mathcal{C}}
\newcommand{\mc}{\mathrm{maxflow}}

\newcommand{\bF}{\mathbf{F}}
\newcommand{\bx}{\mathbf{x}}
\newcommand{\bzero}{\mathbf{0}}
\newcommand{\by}{\mathbf{y}}
\newcommand{\bz}{\mathbf{z}}
\newcommand{\fin}{\bar{F}}
\newcommand{\fout}{F}
\newcommand{\tr}{\top}
\newcommand{\bigO}{\mathcal{O}}

\newcommand{\idxset}{\mathcal{L}}

\newcommand{\bA}{\mathbf{A}}

\newcommand{\badxz}{\Sigma}

\begin{document}
\title{Refined Coding Bounds and Code Constructions for Coherent Network Error Correction}

\author{Shenghao~Yang, 
Raymond~W.~Yeung,~\IEEEmembership{Fellow,~IEEE}, Chi~Kin~Ngai
\thanks{Shenghao Yang, Raymond W. Yeung and Chi~Kin~Ngai are with the
  Department of Information Engineering, The Chinese University of
  Hong Kong.
  Emails:~shenghao.yang@gmail.com,~whyeung@ie.cuhk.edu.hk,~ckngai@alumni.cuhk.net}
}

\maketitle

\begin{abstract}
Coherent network error correction is the error-control problem in network coding with the knowledge of the network codes at the source and sink nodes. 
With respect to a given set of local encoding kernels defining a linear network code, we obtain
refined versions of the Hamming bound, the Singleton bound and the Gilbert-Varshamov bound for coherent network error correction. Similar to its classical counterpart, this refined Singleton bound is tight for linear network codes. The tightness of this refined bound is shown by two construction algorithms of linear network codes achieving this bound. 
These two algorithms illustrate different design methods: one makes use of existing network coding algorithms for error-free transmission and the other makes use of classical error-correcting codes.
The implication of the tightness of the refined Singleton bound is that the sink nodes with higher maximum flow values can have higher error correction capabilities.
\end{abstract}

\begin{keywords} Network error correction, network coding, Hamming bound, Singleton bound, Gilbert-Varshamov bound, network code construction.
\end{keywords}

\section{Introduction}
\label{sec:intr}

Network coding has been extensively studied for multicasting information in a directed communication network when the communication links in the network are error free.
It was shown by Ahlswede \emph{et al.} \cite{flow} that the network capacity for multicast satisfies the max-flow min-cut theorem, and this capacity can be achieved by network coding.
Li, Yeung, and Cai \cite{linear} further showed that it is sufficient to consider linear network codes only.
Subsequently, Koetter and M\'{e}dard \cite{alg} developed a matrix framework for network coding.
Jaggi \emph{et al.} \cite{poly} proposed a deterministic polynomial-time algorithm to construct linear network codes.
Ho \emph{et al.} \cite{ho06j} showed that optimal linear network codes can be efficiently constructed by a randomized algorithm
with an exponentially decreasing probability of failure.

\subsection{Network Error Correction}

Researchers also studied how to achieve reliable
communication by network coding when the communication links are not
perfect. For example, network transmission may suffer from
link failures \cite{alg}, random errors \cite{nwc_err} and
maliciously injected errors \cite{jaggi08j}. 
We refer to these distortions in network transmission collectively as \emph{errors}, and the network coding techniques for combating errors as \emph{network error correction}.

Fig.~\ref{fig:8gaoq} shows one special case of network error correction with two nodes, one source node and one sink node, which are connected by parallel links. This is the model studied in classical algebraic coding theory \cite{macwilliams78,lin04}, a very rich research field for the past 50 years. 

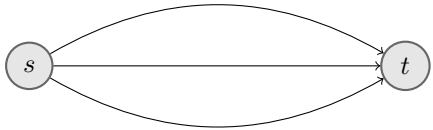
\begin{figure}
  \centering 
  \tikzstyle{dot}=[circle,draw=black!60,fill=gray!20,thick,inner sep=4pt,minimum size=10pt]
  \begin{tikzpicture}[scale=2.5]
    \node[dot] (s) at (-1,0) {$s$};
    \node[dot] (t) at (1,0) {$t$};
    \draw[->] (s) to (t);
    \draw[->] (s) to[out=30,in=150] (t);
    \draw[->] (s) to[out=-30,in=-150] (t);
  \end{tikzpicture}
  \caption[A classical error correction example]{This is a classical error correction example, where $s$ is the source node and $t$ is the sink node. This model is extensively studied by algebraic coding.} \label{fig:8gaoq}
\end{figure}

Cai and Yeung \cite{nwc_err, nec1, nec2} extended the study
of algebraic coding from classical error correction to
network error correction.  They generalized the Hamming
bound (sphere-packing bound), the Singleton bound and the
Gilbert-Varshamov bound (sphere-covering bound) in classical
error correction coding to network coding. 
Zhang studied network error correction in packet networks \cite{zhang08},
where an algebraic definition of the minimum distance for
linear network codes was introduced and the decoding problem
was studied.  
The relation between network codes and
maximum distance separation (MDS) codes in classical
algebraic coding \cite{singleton64} was clarified in
\cite{NWC_theory}.

In \cite{nwc_err, nec1, nec2}, the common assumption is that the sink
nodes know the network topology as well as the network code used in
transmission. This kind of network error correction is referred to as
\emph{coherent network error correction}.  By contrast, network error
correction without this assumption is referred to as \emph{noncoherent
  network error correction}.\footnote{Coherent and noncoherent
  transmissions for network coding are analogous to coherent and
  noncoherent transmissions for multiple antenna channels in wireless
  communications.}  When using the deterministic construction of
linear network codes \cite{linear, poly}, the network transmission is
usually regarded as ``coherent''.  For random network coding, the
network transmission is usually regarded as ``noncoherent''.  It is
possible, however, to use noncoherent transmission for
deterministically constructed network codes and use coherent
transmission for randomly constructed network codes.

In \cite{yang08eu}, Yang \etal\ developed a framework for
characterizing error correction/detection capabilities of
network codes for coherent network error correction.  Their
findings are summarized as follows.  First, the error
correction/detection capabilities of a network code are
completely characterized by a two-dimensional region of
parameters which reduces to the minimum Hamming distance when
1) the network code is linear, and 2) the weight measure on
the error vectors is the Hamming weight.  For a nonlinear
network code, two different minimum distances are needed for
characterizing the capabilities of the code for error
correction and for error detection. This led to the
discovery that for a nonlinear network code, the number of
correctable errors can be more than half of the number of
detectable errors. (For classical algebraic codes, the
number of correctable errors is always the largest integer
not greater than half of the number of detectable errors.)
Further, for the general case, an equivalence relation on
weight measures was defined and it was shown that weight
measures belonging to the same equivalence class lead to the
same minimum weight decoder.  In the special case of network
coding, four weight measures, including the Hamming weight
and others that have been used in various works
\cite{zhang08, adv, weight}, were proved to be in the same
equivalence class for linear network codes.

Network error detection by random network coding has been
studied by Ho \emph{et al.} \cite{byz}.  Jaggi \etal\
\cite{adv,jaggi08j, eav} have developed random algorithms
for network error correction with various assumptions on the
adversaries.  
A part of the work by Zhang \cite{zhang08}
considers packet network error correction when the network
code is not known by receivers, where a sufficient condition
for correct decoding was given in terms of the minimum
distance.  
The distribution of the minimum distance when
applying random network coding was bounded by Balli, Yan and
Zhang \cite{balli07}.  They also studied decoding network
error-correcting codes beyond the error correction
capability \cite{yan07}.

Koetter and Kschischang \cite{koetter08j} introduced a general
framework for noncoherent network error correction. In their
framework, messages are modulated as subspaces, so a code for
noncoherent network error correction is also called a subspace
code. They
proved a Singleton bound, a sphere-packing bound and a sphere-covering
bound for subspace codes.  Using rank-metric codes, Silva and
Kschischang \cite{silva08j} constructed nearly optimal subspace codes
and studied the decoding algorithms.

\subsection{Paper Outline}

In this paper, we follow the framework provided in \cite{yang08eu} to
study the design of linear network codes for coherent
network error correction.

The coding bounds for coherent network error correction
obtained in \cite{nwc_err, nec1,nec2} take only one sink
node with the smallest maximum flow from the source node
into consideration. We observe that each sink node can be
considered individually and a sink node with larger
maximum flow can potentially have higher error
correction/detection capability. These observations lead to
the refined versions of the Hamming bound, the Singleton
bound and the Gilbert-Varshamov bound for network error
correction to be proved in this work.  By way of the weight
properties of network coding, the proof of these bounds are
as transparent as their classical counterparts for
linear network codes.  By contrast, the proofs of the
original versions of these bounds (not necessarily for
linear network codes) in \cite{nec1,nec2} are considerably
more complicated.  The refined Singleton bound was also
implicitly obtained by Zhang \cite{zhang08} independently.
When applying to classical error correction, these bounds
reduce to the classical Hamming bound, the classical
Singleton bound and the classical Gilbert-Varshamov bound,
respectively.

Similar to its classical counterpart, this refined Singleton bound is
tight for linear network codes. The tightness of this refined bound is
shown by two construction algorithms of linear network codes achieving
the bound.  A linear network code consists of two parts, a codebook
and a set of local encoding kernels (defined in Section
\ref{sec:prob-linear}).  Our first algorithm finds a codebook based on
a given set of local encoding kernels.  The set of local encoding
kernels that meets our requirement can be found by the polynomial-time
algorithm in \cite{poly}.  The second algorithm finds a set of local
encoding kernels based on a given classical error-correcting code
satisfying a certain minimum distance requirement as the codebook.
These two algorithms illustrate different design methods.  The set of
local encoding kernels determines the transfer matrices of the
network.  The first algorithm, similar to the classical algebraic
coding, designs a codebook for the transfer matrices.  The second algorithm,
instead, designs transfer matrices to match a codebook.

Various parts of this paper have appeared in \cite{cons,
  refined}.  Subsequent to \cite{cons}, based on the idea of
static network codes \cite{alg}, Matsumoto
\cite{Matsumoto2007} proposed an algorithm for constructing
linear network codes achieving the refined Singleton bound.
In contrast to ours, Matsumoto's algorithm designs the codebook
and the local encoding kernels together.  The complexity and
field size requirements of these three algorithms are
compared.

This paper is organized as follows. In Section
\ref{sec:prob-linear}, we formulate the network error
correction problem and review some previous works.  The
refined coding bounds for coherent network error correction
are proved in Section~\ref{sec:bounds}. In
Section~\ref{sec:tightness}, the tightness of the refined
Singleton bound is proved, and the first construction
algorithm is given. In Section~\ref{sec:algorithms}, we
introduce another construction algorithm that can achieve
the refined Singleton bound.  In the last section, we
summarize our work and discuss future work.

\section{Network Error-Correcting Problem}
\label{sec:prob-linear}

\subsection{Problem Formulation}

Let $\ffield$ be a finite field with $q$ elements. Unless otherwise
specified, all the algebraic operations in this paper are over this
field.  A communication network is represented by a directed acyclic
graph (DAG). (For a comprehensive discussion of directed acyclic
graph, please refer to \cite{wikidag} and the references therein.) A
DAG is an ordered pair $\graph=(\nodeset,\edgeset)$ where $\nodeset$
is the set of nodes and $\edgeset$ is the set of edges.  There can be
multiple edges between a pair of nodes, each of which represents a
communication link that can transmit one symbol in the finite field
$\ffield$.  For an edge $e$ from node $a$ to $b$, we call $a$ ($b$)
the tail (head) of the edge, denoted by $\tail(e)$ ($\head(e)$).  Let
$\In(a)=\{e \in \edgeset : \head(e) = a \}$ and $\Out(a)=\{e \in
\edgeset : \tail(e) = a \}$ be the sets of incoming edges and outgoing
edges of node $a$, respectively.

A directed path in $\graph$ is a sequence of edges
$\{e_i\in \edgeset :i=1,2,\cdots,k\}$ such that 
$\head(e_i)=\tail(e_{i+1})$ for $i=1,2,\cdots,k-1$.
Such a directed path is also called a path from $\tail(e_1)$
to $\head(e_k)$.
A directed acyclic graph gives rise to a partial order
$\leq$ on its nodes, where $a\leq b$ when there
exists a directed path from $a$ to $b$ in the
DAG. Similarly, a DAG gives rise to a partial order $\leq$ on
the edges, where $e \leq e'$ when $e=e'$ or  $\head(e)\leq
\tail(e')$. In other word, $e\leq e'$ if there exists a
directed path from $\tail(e)$ to $\head(e')$ that uses both
$e$ and $e'$. We call this partial order on the edges the \emph{
associated partial order} on the edges.  We extend the
associated partial order on the edges to a total order on
the edges such that for all $e$ and $e'$ in $\edgeset$,
either $e\leq e'$ or $e'\leq e$. Such an extension is not
unique, but we fix one in our discussion and write
$\edgeset=\{e_i:i=1,2,\cdots, |\edgeset|\}$.

A \emph{multicast network} is an ordered triple $(\graph,s,\sinkset)$
where $\graph$ is the network, $s \in \nodeset$ is the source node and
$\sinkset\subset \nodeset$ is the set of sink nodes.  The source node
contains the messages that are demanded by all the sink nodes.
Without loss of generality (WLOG), we assume $\In(s)=\emptyset$. Let
$n_s=|\Out(s)|$. The source node $s$ encodes its message into a row
vector $\mathbf{x}=[x_e, e\in \Out(s)] \in \ffield^{n_s}$, called the
\textit{codeword}.  The set of all codewords is the {\em codebook},
denoted by $\mathcal{C}$.  Note that we do not require $\mathcal{C}$
to be a subspace.  The source node $s$ transmits a codeword by
mapping its $n_s$ components onto the edges in $\Out(s)$.  For any
node $v\neq s$ with $\In(v)=\emptyset$, we assume that this node
outputs the zero element of $\ffield$ to all its outgoing edges.

An {\em error vector} $\mathbf{z}$ is an $|\edgeset|$-dimensional row
vector over $\ffield$ with the $i$th component representing the error
on the $i$th edge in $\edgeset$. An {\em error pattern} is a subset of
$\edgeset$. Let $\rho_{\bz}$ be the error pattern corresponding
to the non-zero components of error vector~$\mathbf{z}$.
An error vector $\bz$ is said to \emph{match} an error
pattern $\rho$ if $\rho_{\bz}\subset \rho$. 
The set of all error vectors
that {match} error pattern $\rho$ is denoted by
$\rho^*$. 
Let $\fin_e$ and $\fout_e$ be the input and output of edge
$e$, respectively, and let the error on the edge be $z_e$. 
The relation between $\fout_e$, $\fin_e$ and $z_e$ is given by
\begin{equation}
  \label{eq:39}
  \fout_e=\fin_e + z_e.
\end{equation}

For any set of edges $\rho$, form two row vectors
\begin{equation*}
 \fout_{\rho} =[\fout_e,e\in \rho],
\end{equation*}
and
\begin{equation*}
 \fin_{\rho}  =[\fin_e,e\in \rho].
\end{equation*}
A network code on network
$\mathcal{G}$ is a codebook $\mathcal{C}\subseteq
\mathbb{F}^{n_s}$ and a family of local encoding functions
$\{\bar{\beta}_e:e\in\edgeset\setminus \Out(s)\}$, where
$\bar{\beta}_e:\ffield^{|\In(\tail(e))|} \rightarrow\ffield$, such
that
\begin{equation}\label{eq:38zoa}
 \fin_e=\bar{\beta}_e(\fout_{\In(\tail(e))}).
\end{equation}

Communication over the network with the network code defined above is
in an upstream-to-downstream order: a node applies its local encoding
functions only after it receives the outputs from all its incoming
edges. Since the network is acyclic, this can be achieved in
light of the partial order on the nodes.  With $\fin_{\Out(s)}=\bx$
and an error vector $\bz$, the symbol $\fin_e$, $\forall
e\in\edgeset$, can be determined inductively by (\ref{eq:39}) and
(\ref{eq:38zoa}).  When we want to indicate the dependence of $\fin_e$
and $\fout_e$ on $\bx$ and $\bz$ explicitly, we will write them as
$\fin_e(\bx,\bz)$ and $\fout_e(\bx,\bz)$, respectively.

A network code is \emph{linear} 
if $\bar{\beta}_e$ is a linear function for all $e\in
\edgeset\setminus \Out(s)$, 
i.e.,
\begin{equation*}
  \fin_e=\sum_{e'\in \edgeset} \beta_{e',e}\fout_{e'},
\end{equation*}
where $\beta_{e',e}$ is called the
\emph{local encoding kernel} from edge $e'$ to edge $e$.
The local encoding kernel $\beta_{e',e}$ can be non-zero only if
$e'\in \In(\tail(e))$.
Define the $|\edgeset|\times|\edgeset|$ one-step transformation
matrix $\mathbf{K}=[K_{i,j}]$ in network $\mathcal{G}$ as
$K_{i,j}=\beta_{e_i,e_j}$. 
For an acyclic network, $\mathbf{K}^N=\mathbf{0}$ for some positive integer
$N$ (see \cite{alg} and \cite{yangthesis} for details). Define the transfer matrix of the network by
$\bF=(\mathbf{I}-\mathbf{K})^{-1}$ \cite{alg}.

For a set of edges $\rho$, define a $|\rho| \times |\edgeset|$  matrix 
$\mathbf{A}_{\rho}=[A_{i,j}]$  by
\begin{equation}\label{eq:a}
  A_{i,j}=\left\{\begin{array}{ll}1& \text{if}\ e_j\ \mathrm{is\
        the}\ i\mathrm{th\ edge\ in}\ \rho, \\ 0 & \mathrm{otherwise.}
    \end{array}
  \right.
\end{equation}
By applying the order on $\edgeset$ to $\rho$, the $|\rho|$ nonzero
columns of $\mathbf{A}_{\rho}$ form an identity matrix. 
To simplify notation, we write $\bF_{\rho,\rho'}=\mathbf{A}_{\rho}\bF\mathbf{A}_{\rho'}^{\top}$.
For input $\bx$ and error vector $\bz$, the output of the
edges in $\rho$ is
\begin{align}  
  \fout_{\rho}(\bx,\bz) & = (\bx \mathbf{A}_{\Out(s)}+\bz)\bF
  \mathbf{A}_{\rho}^{\top} \label{eq:ak78a}\\
   & = \bx \bF_{\Out(s),\rho} + \bz\bF\mathbf{A}_{\rho}^{\top} . \label{eq:467md}
\end{align}

Writing $F_v(\bx,\bz)=\fout_{\In(v)}(\bx,\bz)$ for a node $v$, 
the received vector for a sink node $t$ is
\begin{align}
  \fout_t(\bx,\bz)  =  \mathbf{x}\bF_{s,t}+\mathbf{z}\bF_t, \label{eq:4343a}
\end{align}
where $\bF_{s,t}=\bF_{\Out(s),\In(t)}$, and
$\bF_t=\bF \mathbf{A}_{\In(t)}^{\top}$. Here $\bF_{s,t}$ and $\bF_t$
are the transfer
matrices for message transmission and error transmission, respectively.


\subsection{An Extension of Classical Error Correction}

In this paper, we study error correction coding over the channel given
in \eqref{eq:4343a}, in which $\bF_{s,t}$ and $\bF_{t}$ are known by
the source node $s$ and the sink node $t$.  The channel transformation
is determined by the transfer matrices. In classical error correction
given in Fig.\ref{fig:8gaoq}, the transfer matrices are identity
matrices.  Thus, linear network error correction is an extension of classical
error correction with general transfer matrices.  Our work
follows this perspective to extend a number of results in classical error
correction to network error correction.

Different from classical error correction, network error correction
provides a new freedom for coding design---
the local encoding kernels can be chosen under the constraint of the
network topology. 
One of
our coding algorithm in this paper makes use of this freedom.



\subsection{Existing Results}

In \cite{yang08eu}, Yang \etal\ developed a framework for
characterizing error correction/detection capabilities of linear
network codes for coherent network error correction. They define
equivalence classes of weight measures on error vectors. Weight
measures in the same equivalence class have the same characterizations
of error correction/detection capabilities and induce the same minimum
weight decoder. Four weight measures, namely the Hamming weight and
the others that have been used in the works \cite{zhang08, adv, weight},
are proved to be in the same equivalence class for linear network
codes.  Henceforth, we only consider the Hamming
weight on error vectors in this paper. For sink node $t$ and
nonnegative integer $c$, define 
\begin{equation}\label{eq:deaofae}
 \Phi_t(c) = \{\bz\bF_t:\bz\in\ffield^{|\edgeset|},\ w_H(\bz)\leq c\},
\end{equation}
where $w_H(\bz)$ is the Hamming weight of a vector $\bz$.

\begin{definition}\label{def:3a9176q4}
Consider a linear network code with codebook $\msgset$. For each sink node $t$, define the \emph{distance measure} 
\begin{align}
  D_t(\bx_1,\bx_2) = \min\{c: (\bx_1-\bx_2)\bF_{s,t}
  \in\Phi_t(c)\} \label{eq:dt} 
\end{align}
and define the \emph{minimum distance} of the codebook
\begin{equation}\label{eq:def8af}
 d_{\min,t} = \min_{\bx_1\neq\bx_2\in\msgset} D_t(\bx_1,\bx_2).
\end{equation}
\end{definition}

We know that $D_t$ is a translation-invariant metric
\cite{yang08eu}.
Consider $\bx_1, \bx_2\in\msgset$.
For any $\bz$ with $\bz\bF_t=(\bx_1-\bx_2)\bF_{s,t}$, we have
$(\bx_1-\bx_2)\bF_{s,t}\in \Phi_t(w_H(\bz))$. Thus $$D_t(\bx_1,\bx_2)
\leq \min_{\bz:(\bx_1-\bx_2)\bF_{s,t} = \bz\bF_t} w_H(\bz).$$
On the other hand, we see that $(\bx_1-\bx_2)\bF_{s,t}
  \in\Phi_t(D_t(\bx_1,\bx_2))$. So, there exists $\bz\in\ffield^{|\edgeset|}$ with $w_H(\bz)=D_t(\bx_1,\bx_2)$ and $(\bx_1-\bx_2)\bF_{s,t}=\bz\bF_t$. Thus, 
$$D_t(\bx_1,\bx_2)
\geq \min_{\bz:(\bx_1-\bx_2)\bF_{s,t} = \bz\bF_t} w_H(\bz).$$
Therefore, we can equivalently write
\begin{equation}
  D_t(\bx_1,\bx_2) = \min_{\bz:(\bx_1-\bx_2)\bF_{s,t} = \bz\bF_t}
  w_H(\bz). \label{eq:dt2}
\end{equation}

\begin{definition}
  \emph{Minimum Weight Decoder I} at a sink node $t$, denoted by
  $\mwd_t^I$, decodes a received vector $\by$ as follows: First, find
  all the solutions of the equation
\begin{equation}\label{eq:ge:38g}
  F_t(\bx,\bz)= \by
\end{equation}
with $\bx\in\msgset$ and $\bz\in\errset$ as variables.  A pair $(\bx,
\bz)$, consisting of the message part $\bx$ and the error part $\bz$,
is said to be a solution if it satisfies (\ref{eq:ge:38g}), and
$(\bx,\bz)$ is a minimum weight solution if $w_H(\bz)$ achieves the
minimum among all the solutions. 
If all the minimum weight solutions
have the identical message parts, the decoder outputs the common
message part as the decoded message. Otherwise, the decoder outputs a
warning that errors have occurred.
\end{definition}

A code is \emph{$c$-error-correcting} at sink node $t$ if all error
vectors $\bz$ with $w_H(\bz)\leq c$ are correctable by $\mwd_t^I$.

\begin{theorem}[\cite{yang08eu}]\label{the:2}
 A linear network code is $c$-error-correcting at sink node $t$ if and only if $d_{\min,t}\geq 2c+1$.
\end{theorem}

For two subsets $V_1,V_2\subset \mathbb{F} ^{n_s}$, define
\begin{equation*}
	V_1+V_2=\{\mathbf{v}_1+\mathbf{v}_2:\mathbf{v}_1\in
	V_1,\mathbf{v}_2 \in V_2\}.
\end{equation*}
For $\mathbf{v}\in \mathbb{F} ^{n_s}$ and $V \subset \mathbb{F}
^{n_s}$, we also write $\{\mathbf{v}\}+V$ as $\mathbf{v}+V$.  For
sink node $t$ and nonnegative integer $c$, define the \emph{decoding
  sphere} of a codeword $\mathbf{x}$ as
\begin{align} 
  \Phi_t(\mathbf{x},c) & = \{F_t(\bx,\bz): \bz
  \in\ffield^{|\edgeset|}, w_H(\bz)\leq c\}\nonumber \\
   & = \bx \bF_{s,t} + \Phi_t(c)\label{eq:28dfq}
\end{align}

\begin{definition}
  If $\Phi_t(\bx,c)$ for all $\bx\in \msgset$ are nonempty and
  disjoint, \emph{Minimum Weight Decoder II} at sink node $t$, denoted
  by $\mwd^{II}_t(c)$, decodes a received vector $\by$ as follows: If
  $\by \in \Phi_t(\bx, c)$ for some $\bx\in \msgset$, the decoder
  outputs $\bx$ as the decoded message. If $\by$ is not in any of the
  decoding spheres, the decoder outputs a warning that errors have
  occurred.
\end{definition}

A code is \emph{$c$-error-detecting} at sink node $t$ if
$\mwd_t^{II}(0)$ exists and all error vector $\bz$ with $0<
w_H(\bz)\leq c$ are detectable by $\mwd_t^{II}(0)$.

\begin{theorem}[\cite{yang08eu}]\label{the:4398zk}
  A code is $c$-error-detecting at sink node $t$ if and only if
  $d_{\min,t}\geq c+1$.
\end{theorem}

Furthermore, we can use $\mwd_t^{II}(c)$, $c>0$, for joint error
correction and detection. Erasure correction is error correction with
the potential positions of the errors in the network known by the
decoder.  We can similarly characterize the erasure correction
capability of linear network codes by $d_{\min,t}$. Readers are
referred to \cite{yang08eu} for the details.

There exist coding bounds on network codes that corresponding to the
classical Hamming bound, Singleton bound and
Gilbert-Varshamov bound. We review some of the results in
\cite{nec1, nec2}. 
The \emph{maximum flow} from node $a$ to node $b$
is the maximum number of edge-disjoint paths from $a$ to $b$, denoted
by $\text{maxflow}(a,b)$.
Let
\begin{equation*}
 d_{\min}=\min_{t\in\mathcal{T}}d_{\min,t},
\end{equation*}
and
\begin{equation*}
 n=\min_{t\in\mathcal{T}}\mathrm{maxflow}(s,t).
\end{equation*}
In terms of the notion of minimum distance, the Hamming bound and the
Singleton bound for network codes obtained in
\cite{nec1} can be restated as
\begin{equation}\label{eq:hamming1}
  |\mathcal{C}|\leq \frac{q^{n}}{\sum_{i=0}^{\tau}\binom{n}{i}(q-1)^i},
\end{equation}
where $\tau=\lfloor \frac{d_{\min}-1}{2}  \rfloor$,
and
\begin{equation}\label{eq:singleton1}
  |\mathcal{C}|\leq q^{n-d_{\min} + 1},
\end{equation}
respectively, where $q$ is the field size.
The tightness of (\ref{eq:singleton1}) has been proved in \cite{nec2}.

\section{Refined Coding Bounds}\label{sec:bounds}

In this section, we present refined versions of the coding
bounds in \cite{nec1,nec2} for linear network codes. In
terms of the distance measures developed in \cite{yang08eu},
the proofs of these bounds are as transparent as the their
classical counterparts.

\subsection{Hamming Bound and Singleton Bound}

\begin{theorem}\label{the:hamming} 
  Consider a linear network code with codebook $\mathcal{C}$,
  $\rank(\bF_{s,t})=r_t$ and $d_{\min,t}>0$. Then $|\msgset|$
  satisfies
  \begin{enumerate} 
  \item the refined Hamming bound 
    \begin{equation}\label{eq:hamming} 
    |\mathcal{C}|\leq \min_{t\in\sinkset}\frac{q^{r_t}}{\sum_{i=0}^{\tau_t}\binom{r_t}{i}(q-1)^i}, 
    \end{equation} 
  where $\tau_t=\lfloor \frac{d_{\min,t}-1}{2}  \rfloor$, 
  and
  \item the refined Singleton bound 
  \begin{equation}\label{eq:singleton} 
    |\mathcal{C}|\leq q^{r_t-d_{\min,t}+1}, 
  \end{equation} 
  for all sink nodes $t$.  
  \end{enumerate} 
\end{theorem} 
 
\textbf{Remark:} The refined Singleton bound can be rewritten as  
\begin{equation*} 
 d_{\min,t} \leq r_t - \log_q |\mathcal{C}| +1 \leq
 \mathrm{maxflow}(s,t) - \log_q |\mathcal{C}| +1,  
\end{equation*} 
for all sink nodes $t$, which suggests that the sink nodes with larger
maximum flow values can potentially have higher error correction
capabilities. We present network codes that achieve this bound in
Section \ref{sec:tightness} and \ref{sec:algorithms}.

\begin{proof} 
  Fix a sink node $t$.  Since $\rank(\bF_{s,t})=r_t$, we can find
  $r_t$ linearly independent rows of $\bF_{s,t}$.  Let $\rho_t \subset
  \Out(s)$ such that $|\rho_t|=r_t$ and $\bF_{\rho_t,\In(t)}$ is a
  full rank submatrix of $\bF_{s,t}$.  Note that $\rho_t$ can be
  regarded as an error pattern.  Define a mapping
  $\phi_t:\mathcal{C}\rightarrow \ffield^{r_t}$ by
  $\phi_t(\mathbf{x})=\mathbf{x}'$ if
  $\mathbf{x}'\bF_{\rho_t,\In(t)}=\mathbf{x}\bF_{s,t}$.  Since the
  rows of $\bF_{\rho_t,\In(t)}$ form a basis for the row space of
  $\bF_{s,t}$, $\phi_t$ is well defined.  The mapping $\phi_t$ is
  one-to-one because otherwise there exists $\mathbf{x}' \in
  \ffield^{r_t}$ such that
  $\mathbf{x}'\bF_{\rho_t,\In(t)}=\mathbf{x}_1 \bF_{s,t} =
  \mathbf{x}_2 \bF_{s,t}$ for distinct $\mathbf{x}_1, \mathbf{x}_2 \in
  \mathcal{C}$, a contradiction to the assumption that $d_{\min,t}>0$.
  Define
  \begin{align*} 
   \mathcal{C}_t & =\{\phi_t(\bx):\bx\in \mathcal{C}\}.
  \end{align*} 
  Since $\phi_t$ is a   one-to-one mapping,
   $|\mathcal{C}_t|=|\mathcal{C}|$. 
 
  We claim that, as a classical error-correcting code of length
  $r_t$, $\mathcal{C}_t$ has minimum distance
  $d_{\min}(\mathcal{C}_t) \geq d_{\min,t}$. 
  We prove this claim by contradiction. If $d_{\min}(\mathcal{C}_t)
  < d_{\min,t}$, it means there exist
  $\bx_1',\bx_2'\in\mathcal{C}_t$ such that
  $w_H(\bx_1'-\bx_2')<d_{\min,t}$. 
  Let  $\bx_1 =\phi_t^{-1}(\bx_1')$ and $\bx_2 =
  \phi_t^{-1}(\bx_2')$. We know that $\bx_1,\bx_2\in\mathcal{C}$, and
  \begin{align*}
    (\bx_1-\bx_2)\bF_{s,t} = (\bx_1'-\bx_2')\bF_{\rho_t,\In(t)} =
     \bz'\bF_{t},
  \end{align*}
  where $\bz'=(\bx_1'-\bx_2')\mathbf{A}_{\rho_t}$.
  Thus, 
  \begin{align*}
    D_t(\bx_1,\bx_2)  \leq w_H(\bz') = w_H(\bx_1'-\bx_2') < d_{\min,t},
  \end{align*}
  where the first inequality follows from \eqref{eq:dt2}.  So we have
  a contradiction to $d_{\min,t}\leq D_t(\bx_1,\bx_2)$ and hence
  $d_{\min}(\mathcal{C}_t)\geq d_{\min,t}$ as claimed.  Applying the
  Hamming bound and the Singleton bound for classical error-correcting
  codes to $\mathcal{C}_t$, we have
  \begin{equation*} 
    |\mathcal{C}_t| \leq \frac{q^{r_t}}{\sum_{i=0}^{\tau_t'}\binom{r_t}{i}(q-1)^i}  \leq \frac{q^{r_t}}{\sum_{i=0}^{\tau_t}\binom{r_t}{i}(q-1)^i}, 
  \end{equation*} 
  where $\tau_t'=\lfloor \frac{d_{\min}(\mathcal{C}_t)-1}{2}  \rfloor\geq \tau_t$,
  and 
  \begin{equation*} 
    |\mathcal{C}_t|\leq q^{r_t-d_{\min}(\mathcal{C}_t)+1} \leq q^{r_t-d_{\min,t}+1}. 
  \end{equation*} 
  The proof is completed by noting that
  $|\mathcal{C}|=|\mathcal{C}_t|$.  
\end{proof} 
\textbf{Remark:} Let $f$ be an upper bound on the size of a classical
block code in terms of its minimum distance such that $f$ is
monotonically decreasing. Examples of $f$ are the Hamming bound and
the Singleton bound. Applying this bound to $\mathcal{C}_t$, we have
\begin{equation*}
  |\mathcal{C}_t| \leq f(d_{\min}(C_t)).
\end{equation*}
Since $f$ is monotonically decreasing, together with
$d_{\min}(\mathcal{C}_t) \geq d_{\min,t}$ as shown in the above proof, we have
\begin{equation}
  \label{eq:b2}
  |\mathcal{C}| = |\mathcal{C}_t| \leq f(d_{\min}(C_t)) \leq f(d_{\min,t}).
\end{equation}
In other words, the bounds in \eqref{eq:b2} is simply the upper bound
$f$ applied to $\mathcal{C}$ as if $\mathcal{C}$ is a classical block
code with minimum distance $d_{\min,t}$.

 
\begin{lemma}\label{lemma:194kd} 
\begin{equation*} 
  \frac{q^{m}}{\sum_{i=0}^{\tau}\binom{m}{i}(q-1)^i} < 
  \frac{q^{m+1}}{\sum_{i=0}^{\tau}\binom{m+1}{i}(q-1)^i} 
\end{equation*} 
for $\tau\leq m / 2$. 
\end{lemma} 
\begin{proof} 
 This inequality can be established by considering 
\begin{eqnarray} 
  \frac{q^{m}}{\sum_{i=0}^{\tau}\binom{m}{i}(q-1)^i} & = & 
  \frac{q^{m+1}}{\sum_{i=0}^{\tau}\frac{q(m-i+1)}{m+1}\binom{m+1}{i}(q-1)^i} \nonumber
  \\ &<& \frac{q^{m+1}}{\sum_{i=0}^{\tau}\binom{m+1}{i}(q-1)^i} 
  \label{eq:pf3}, 
\end{eqnarray} 
where (\ref{eq:pf3}) holds because $\frac{q(m-i+1)}{m+1}> 1$ given that $q\geq 
2$ and $i \le \tau \le m/2$. 
\end{proof} 
 

The refined Hamming bound and the refined Singleton bound, as we will
show, imply the bounds shown in (\ref{eq:hamming1}) and
(\ref{eq:singleton1}) but not vice versa.  The refined Hamming bound
implies
  \begin{align} 
    |\mathcal{C}| & \leq
    \frac{q^{r_t}}{\sum_{i=0}^{\tau_t}\binom{r_t}{i}(q-1)^i} \nonumber\\  
    & \leq  
    \frac{q^{r_t}}{\sum_{i=0}^{\tau}\binom{r_t}{i}(q-1)^i}\label{eq:hammid} \\ 
    & \leq  \frac{q^{\mathrm{maxflow}(s,t)}} 
    {\sum_{i=0}^{\tau}\binom{\mathrm{maxflow}(s,t)}{i}(q-1)^i} 
    \label{eq:hammid2} 
  \end{align} 
  for all sink nodes $t$, where (\ref{eq:hammid}) follows from $\tau =
  \lfloor \frac{d_{\min}-1}{2}\rfloor \leq \lfloor
  \frac{d_{\min,t}-1}{2} \rfloor = \tau_t$, and (\ref{eq:hammid2})
  follows from $r_t\leq\mathrm{maxflow}(s,t)$ and the inequality
  proved in Lemma \ref{lemma:194kd}.  By the same inequality, upon
  minimizing over all sink nodes~$t \in \mathcal{T}$, we obtain
  (\ref{eq:hamming1}).  Toward verifying the condition for applying
  the inequality in Lemma~\ref{lemma:194kd} in the above, we see $r_t
  \geq d_{\min,t}-1$ since $1\leq |\mathcal{C}|\leq
  q^{r_t-d_{\min,t}+1}$. Then
\begin{align*} 
\tau \leq  \tau_t \leq  \frac{ d_{\min,t} - 1}{ 2 }  
 \leq  \frac{ r_t}{2} 
\end{align*} 
for all $t \in \mathcal{T}$.

The refined Singleton bound is maximized when $r_t =
\mathrm{maxflow}(s,t)$ for all $t \in \mathcal{T}$.  This can be
achieved by a \emph{linear broadcast code} whose existence was proved in
\cite{linear}, \cite{NWC_theory}.  To show that the refined 
Singleton bound implies (\ref{eq:singleton1}), 
consider 
\begin{align*} 
  |\mathcal{C}| & \leq  q^{r_t-d_{\min,t}+1}\\ & \leq  
q^{r_t-d_{\min}+1} \\ & \leq  q^{\mathrm{maxflow}(s,t)-d_{\min}+1} 
\end{align*} 
for all sink nodes $t$.  Then (\ref{eq:singleton1}) is obtained upon 
minimizing over all $t \in \mathcal{T}$.

\subsection{Sphere-Packing Bound} 
 
For nonnegative integer $d$, define
  \begin{equation} \label{eq:q8vdelta} 
    \Delta_t(\mathbf{x},d)=\{\mathbf{x}'\in 
    \ffield^{n_s}:D_t(\mathbf{x}',\mathbf{x})\leq d\}. 
  \end{equation} 
Here $D_t(\cdot,\cdot)$ is defined in \eqref{eq:dt}.
Since $D_t$ is a translation invariant metric \cite{yang08eu}, we have 
$\Delta_t(\mathbf{x},d) = \bx+ \Delta_t(\bzero,d)$, which implies 
$|\Delta_t(\mathbf{x},d)|=|\Delta_t(\bzero,d)|$. 
Another fact is that  $\Delta_t(\mathbf{0},d)$ 
is closed under scalar multiplication, i.e., 
\begin{equation*} 
 \alpha\Delta_t(\mathbf{0},d)\triangleq \{\alpha\mathbf{x}:\mathbf{x}\in\Delta_t(\mathbf{0},d)\}=\Delta_t(d), 
\end{equation*} 
where $\alpha\in\mathbb{F}$ and $\alpha\neq 0$. 
 
\begin{lemma}\label{lemma:a81ki} 
 \begin{equation} 
  \label{eq:332sgs} 
  \binom{|\edgeset|}{d}q^{d} >
  |\Delta_t(\mathbf{0},d)|q^{-(n_s-r_t)} =|\Phi_t(d)| \geq 
  \sum_{i=0}^{d}\binom{r_t}{i}(q-1)^i, 
\end{equation} 
where $r_t=\rank(\bF_{s,t})$ and $d\leq r_t$. 
\end{lemma} 
\begin{proof} 
  Applying the definition of $D_t$, $\Delta_t(\mathbf{0},d)$ can be
  rewritten as
\begin{align}\label{eq:1kzra} 
   \Delta_t(\mathbf{0},d) & = \{\bx \in \ffield^{n_s}: \bx\bF_{s,t}\in \Phi_t(d)\}, 
\end{align}  
where $\Phi_t$ is defined in (\ref{eq:deaofae}).
Since the rank of $\bF_{s,t}$ is $r_t$, the
null space of $\bF_{s,t}$ defined as 
\begin{equation*}
  \text{Null}(\bF_{s,t}) = \{\bx:\bx \bF_{s,t} = \bzero \}
\end{equation*}
has dimension $n_s-r_t$.  By the theory of linear system of equations,
for each vector $\by$ in $\Phi_t(d)$, we have
$|\text{Null}(\bF_{s,t})| = q^{n_s-r_t}$ vector $\bx$ satisfies $\bx \bF_{s,t}=\by$, and all such $\bx$ are in $\Delta_t(\mathbf{0},d)$. Thus,
\begin{equation} \label{eq:18g940} 
  |\Delta_t(\mathbf{0},d)| = q^{n_s-r_t}|\Phi_t(d)|.
\end{equation} 
By the definition of $\Phi_t$, we have 
\begin{align}
 |\Phi_t(d)| & \leq |\{\bz\in\ffield^{|\edgeset|}:w_H(\bz)\leq d\}| 
 \nonumber \\ & < \binom{|\edgeset|}{d}q^{d}. \label{eq:1auqjaaa} 
\end{align} 
Together with (\ref{eq:18g940}), we obtain the first inequality in
(\ref{eq:332sgs}).

Since $\rank(\bF_{s,t})=r_t$, we can find $r_t$ linearly independent
rows of $\bF_{s,t}$.  Let $\rho_t \subset \Out(s)$ such that $|\rho_t|=r_t$ and
$\bF_{\rho_t,\In(t)}$ is a full row rank submatrix of $\bF_{s,t}$.
Note that 
$\bF_{\rho_t,\In(t)}$ is also a submatrix of $\bF_t$.  Since,
\begin{align*} 
 \Phi_t(d) = \{ \bz\bF_t : w_H(\bz) \leq d \} \supset \{\bz' \bF_{\rho_t,\In(t)}: w_H(\bz')\leq d\},
\end{align*} 
we have
\begin{align*}
  |\Phi_t(d)| & \geq |\{\bz' \bF_{\rho_t,\In(t)}: w_H(\bz')\leq d\}|  \\
  & = |\{\bz'\in\ffield^{r_t}:w_H(\bz')\leq d\}|  \\
  & = \sum_{i=0}^{d}\binom{r_t}{i}(q-1)^i.
\end{align*}
The proof is complete.
\end{proof} 

Using the idea of sphere packing, we have the following stronger 
version of the refined Hamming bound in Theorem~\ref{the:hamming}.  

\begin{theorem}[Sphere-packing bound]\label{the:packing} 
  A linear network code with codebook $\mathcal{C}$ 
  and positive minimum distance $d_{\min,t}$ for all sink nodes $t$ 
  satisfies 
  \begin{align*} 
    |\msgset| \leq \frac{q^{r_t}}{|\Phi_t(\tau_{t})|}, 
  \end{align*} 
  where $\tau_t=\lfloor \frac{d_{\min,t}-1}{2}  \rfloor$.  
\end{theorem} 
\begin{proof} 
  For different codewords $\bx_1$ and $\bx_2$, we show that
  $\Delta_t(\mathbf{x}_1,\tau_{t})$ and
  $\Delta_t(\mathbf{x}_2,\tau_{t})$ are disjoint by contradiction. 
  Let 
  \begin{align*} 
    \bx \in \Delta_t(\mathbf{x}_1,\tau_{t}) \cap \Delta_t(\mathbf{x}_2,\tau_{t}).
  \end{align*} 
  By the definition of $\Delta_t$ in (\ref{eq:q8vdelta}), we have  
  $D_t(\bx_1,\bx)\leq \tau_t$ and $D_t(\bx_2,\bx)\leq \tau_t$.
  Applying the triangle inequality of $D_t$, we have
  \begin{align*} 
   D_t(\bx_1,\bx_2) & \leq D_t(\bx_1,\bx) + D_t(\bx_2,\bx) \\ 
    & \leq 2\tau_t \\ 
    & \leq d_{\min,t}-1, 
  \end{align*} 
  which is a contradiction to the definition of $d_{\min,t}$.
  Therefore, $q^{n_s}\geq \sum_{\bx\in\msgset}|\Delta_t(\mathbf{x},\tau_{t})| 
  =|\msgset||\Delta_t(\bzero,\tau_t)|$.  The proof is
  complete by considering the equality in Lemma~\ref{lemma:a81ki}.
\end{proof} 
 
Applying the second inequality in Lemma~\ref{lemma:a81ki},
Theorem~\ref{the:packing} implies the refined Hamming bound
in Theorem~\ref{the:hamming}. Thus Theorem~\ref{the:packing}
gives a potentially tighter upper bound on $|\msgset|$ than
the refined Hamming bound, although the former is less
explicit than the latter.
 
\subsection{Gilbert Bound and Varshamov Bound} 
 
We have the following sphere-covering type bounds for linear network codes. 
 
\begin{theorem}[Gilbert bound]\label{the:gilbert} 
  Given a set of local encoding kernels, let $|\mathcal{C}|_{\max}$ be 
  the maximum possible size of codebooks such that the network 
  code has positive minimum distance $d_{\min,t}$ for each sink node $t$. 
  Then, 
  \begin{equation}\label{eq:gv} 
    |\mathcal{C}|_{\max} \geq \frac{q^{n_s}}{|\Delta(\mathbf{0})|}, 
  \end{equation} 
  where 
  \begin{equation}\label{eq:delta0} 
    \Delta(\mathbf{0})=\cup_{t\in\sinkset}\Delta_t(\mathbf{0},d_{\min,t}-1). 
  \end{equation} 
 
\end{theorem} 
\begin{proof} 
 Let $\mathcal{C}$ be a codebook with the maximum possible size, and let
 \begin{align*} 
   \Delta(\mathbf{c})=\cup_{t\in\sinkset}\Delta_t(\mathbf{c},d_{\min,t}-1). 
 \end{align*} 
 For any $\mathbf{x}\in\ffield^{n_s}$, 
 there exists a codeword $\mathbf{c} \in \mathcal{C}$ and a sink node 
 $t$ such that 
 \begin{align*} 
  D_t(\mathbf{x},\mathbf{c})\leq d_{\min,t}-1,
 \end{align*} 
 since otherwise we could add $\mathbf{x}$ to the codebook while keeping 
 the minimum distance.
 By definition, we know
 \begin{align*}
   \Delta(\mathbf{c}) = \cup_{t\in\sinkset} \{\mathbf{x}\in 
    \ffield^{n_s}:D_t(\mathbf{x},\mathbf{c})\leq d_{\min,t}-1\}.
 \end{align*}
 Hence, the whole space $\ffield^{n_s}$ is contained in the 
 union of $\Delta(\mathbf{c})$ over all codewords $\mathbf{c} \in 
 \mathcal{C}$, i.e., 
 \begin{align*} 
  \mathbb{F} ^{n_s} = \cup_{\mathbf{c} \in \mathcal{C}}\Delta(\mathbf{c}). 
 \end{align*} 
 Since $\Delta(\mathbf{c})=\mathbf{c}+\Delta(\mathbf{0})$, we have 
 $|\Delta(\mathbf{c})|=|\Delta(\mathbf{0})|$. 
 So we deduce that 
 $q^{n_s}\leq |\mathcal{C}||\Delta(\mathbf{0})|$. 
\end{proof} 


We say a codebook is \emph{linear} if it is a
vector space. 

\begin{lemma}\label{lemma:is}
  Consider a linear network code with linear
codebook $\mathcal{C}$. The minimum distance $d_{\min,t}\geq
d$ if and only if 
\begin{equation*}
  \mathcal{C}\cap \Delta_t(\bzero, d-1) = \{\bzero\}.
\end{equation*}
\end{lemma}
\begin{proof}
If there exists $\bx \in \mathcal{C}\cap \Delta_t(\bzero,
d-1)$ and $\bx \neq \bzero$, then
$D_t(\bzero, \bx)<d$. Since $\bzero\in \mathcal{C}$, we have
$d_{\min,t} <d$. This proves the sufficient condition. 

Now we prove the necessary condition. 
For $\bx_1,\bx_2\in\mathcal{C}$, $\bx_1-\bx_2\in\mathcal{C}$. 
Since
\begin{equation*}
  D_t(\bx_1,\bx_2) = D_t(\bx_1-\bx_2,\bzero),
\end{equation*}
we have
\begin{equation*}
  d_{\min,t}=\min_{\bx\in \mathcal{C},\bx\neq \bzero} D_t(\bx,\bzero).
\end{equation*}
Thus,
\begin{equation*}
  \mathcal{C}\cap \Delta_t(\bzero, d_{\min,t}-1) = \{\bzero\}.
\end{equation*}
The proof is completed noting that
$\Delta_t(\bzero,d_{\min,t}-1)\supset \Delta_t(\bzero,d-1)$.
\end{proof}

\begin{theorem}[Varshamov bound]\label{the:varshamov} 
  Given a set of local encoding kernels, let $\msgdim_{\max}$ be 
  the maximum possible dimension of \emph{linear} codebooks such 
  that the network code has positive minimum distance $d_{\min,t}$ 
  for each sink node $t$. 
  Then, 
  \begin{equation} \label{43niojh} 
	  \msgdim_{\max} \ge n_s - \log_q |\Delta(\mathbf{0})|, 
  \end{equation} 
  where $\Delta(\mathbf{0})$ is defined in (\ref{eq:delta0}). 
\end{theorem} 
\begin{proof} 
  Let $\mathcal{C}$ be a linear codebook with the maximum
  possible dimension. 
  By Lemma~\ref{lemma:is}, $\mathcal{C}\cap \Delta(\bzero) =\{\bzero\}$.
  We claim that
  \begin{equation} 
     \mathbb{F} ^{n_s}=\Delta(\mathbf{0})+\mathcal{C}. 
     \label{eq:claim} 
  \end{equation} 
  If the claim is true, then
  \begin{align*}
    q^{n_s} = |\Delta(\mathbf{0})+ \mathcal{C}| \leq
    |\Delta(\mathbf{0})||\mathcal{C}| =
    |\Delta(\mathbf{0})|q^{\msgdim_{\max}},
  \end{align*}
  proving (\ref{43niojh}).
	 
  Since $\mathbb{F}
  ^{n_s}\supset\Delta(\mathbf{0})+\mathcal{C}$, so we only
  need to show
  $\mathbb{F} ^{n_s}\subset\Delta(\mathbf{0})+\mathcal{C}$. 
  Assume there exists
  \begin{equation}
    \mathbf{g}\in 
    \mathbb{F} ^{n_s}\setminus(\Delta(\mathbf{0})+\mathcal{C}). 
    \label{eq:asg} 
  \end{equation} 
  Let $\mathcal{C}'=\mathcal{C}+ \langle \mathbf{g}
  \rangle$.  Then $\mathcal{C}'$ is a subspace with
  dimension $\msgdim_{\max}+1$.  If
  $\mathcal{C}'\cap\Delta(\mathbf{0})\neq\{\mathbf{0}\}$,
  then there exists a non-zero vector
  \begin{equation} \label{4nlnl} \mathbf{c}+\alpha
    \mathbf{g}\in\Delta(\mathbf{0}),
  \end{equation} 
  where $\mathbf{c}\in \mathcal{C}$ and $\alpha\in
  \mathbb{F}$. Here, $\alpha \ne 0$, otherwise we have
  $\mathbf{c}=\mathbf{0}$ because
  $\mathcal{C}\cap\Delta(\mathbf{0})=\{\mathbf{0}\}$.  Since
  $\Delta_t(\mathbf{0},d_{\min,t}-1)$ is closed under scalar
  multiplication for all $t \in \sinkset$, see from
  (\ref{eq:delta0}) that the same holds for
  $\Delta(\mathbf{0})$.  Thus from (\ref{4nlnl}),
  \begin{equation*}
    \mathbf{g} \in \Delta(\mathbf{0}) -
    \alpha^{-1}\mathbf{c} \subset
    \Delta(\mathbf{0})+\mathcal{C},
  \end{equation*}
  which is a contradiction to (\ref{eq:asg}).  Therefore,
  $\mathcal{C}'\cap\Delta(\mathbf{0})=\{\mathbf{0}\}$.  By
  Lemma~\ref{lemma:is}, $\mathcal{C}'$ is a codebook such
  that the network code has unicast minimum distance larger
  than or equal to $d_{\min,t}$, which is a contradiction on
  the maximality of $\mathcal{C}$.  The proof is completed.
\end{proof} 






\section{Tightness of the Singleton Bound and Code Construction}
\label{sec:tightness}

For an $(\msgdim, (r_t:t\in \sinkset), (d_t:t\in\sinkset))$
linear network code, we refer to one for which the codebook
$\msgset$ is an $\msgdim$-dimensional subspace of
$\ffield^{n_s}$, the rank of the transfer matrix $\bF_{s,t}$
is $r_t$, and the minimum distance for sink node $t$ is at
least $d_t$, $t \in \mathcal{T}$.
In this section, we propose an algorithm to construct
$(\msgdim, (r_t:t\in \sinkset), (d_t:t\in\sinkset))$ linear
network codes that can achieve the refined Singleton bound.

\subsection{Tightness of the Singleton Bound} 
 
\begin{theorem} \label{the:ach_sin} Given a set of local
  encoding kernels with $r_t=\rank(\bF_{s,t})$ over a finite
  field with size $\fsize$, for every
  \begin{equation} 
  0 < \msgdim \leq \min_{t \in \mathcal{T}} r_t, 
  \label{n4onlkn} 
  \end{equation} 
  there exists a codebook $\mathcal{C}$  with
  $|\mathcal{C}|=q^\msgdim$ such that  
  \begin{equation}  
  d_{\min,t}= r_t-\msgdim+1  
  \label{u8bldo} 
  \end{equation} 
   for all sink nodes $t$, provided that $q$ is sufficiently large. 
\end{theorem} 
\begin{proof} 
  Fix an $\msgdim$ which satisfies (\ref{n4onlkn}).  We will
  construct an $\msgdim$-dimensional linear codebook which
  together with the given set of local encoding kernels
  constitutes a linear network code that satisfies
  (\ref{u8bldo}) for all $t$.  Note that (\ref{n4onlkn}) and
  (\ref{u8bldo}) imply
  \begin{equation*}
    d_{\min,t} \geq 1. 
  \end{equation*}

  We construct the codebook $\mathcal{C}$ by finding a
  basis. Let 
  $\mathbf{g}_1,\cdots,\mathbf{g}_\msgdim\in\ffield^{n_s}$ be a 
  sequence of vectors 
  obtained as follows. For each $i$, $1\leq i\leq \msgdim$, choose 
  $\mathbf{g}_i$ such   that 
  \begin{equation} 
	  \mathbf{g}_i\notin \Delta_t(\mathbf{0},r_t- 
          \msgdim)+ \langle \mathbf{g}_1,\cdots,\mathbf{g}_{i-1} \rangle 
	  \label{eq:gti} 
  \end{equation} 
  for each  sink node $t$. 
  As we will show, this implies   
  \begin{equation} 
	  \Delta_t(\mathbf{0},r_t-\msgdim)\cap 
          \langle \mathbf{g}_1,\cdots,\mathbf{g}_{i} \rangle 
	  =\{\mathbf{0}\}	 
	  \label{eq:gtg} 
  \end{equation} 
  for each sink node $t$.  If such
  $\mathbf{g}_1,\cdots,\mathbf{g}_\msgdim$ exist, then we claim that
  $\mathcal{C}= \langle \mathbf{g}_1,\cdots,\mathbf{g}_\msgdim
  \rangle$ is a codebook with the desired properties.  To verify this
  claim, first, we see that $\mathbf{g}_1,\cdots,\mathbf{g}_\msgdim$
  are linearly independent since (\ref{eq:gti}) holds for $i=1,\cdots,
  \msgdim$; second, we have $d_{\min,t}\geq r_t-\msgdim+1$ since
  (\ref{eq:gtg}) holds for $i=\msgdim$ (ref Lemma \ref{lemma:is}).
  Note that by (\ref{eq:singleton}), the refined Singleton bound, we
  indeed have $d_{\min,t}= r_t-\msgdim+1$, namely (\ref{u8bldo}) for
  any sink node $t$.

  Now we show that $\mathbf{g}_i$ satisfying (\ref{eq:gti}) exists 
  if the field size $q$ is sufficiently large. Observe that 
  \begin{align} 
    \lefteqn{|\Delta_t(\mathbf{0},r_t-\msgdim)+ 
      \langle \mathbf{g}_1,\cdots,\mathbf{g}_{i-1} \rangle 
      |}\nonumber \\ 
    & \leq |\Delta_t(\mathbf{0},r_t-\msgdim)|q^{i-1}\nonumber \\ 
    & \leq \binom{|\edgeset|}{r_t-\msgdim}q^{r_t-\msgdim}q^{n_s-r_t}q^{i-1} \label{eq:427f0sg}\\ 
    & = \binom{|\edgeset|}{r_t-\msgdim}q^{n_s-\msgdim+i-1}, \nonumber
  \end{align} 
  where (\ref{eq:427f0sg}) follows from Lemma
  \ref{lemma:a81ki}. 
  If
  \begin{equation}\label{eq:fll}
    q^{n_s} > \sum_{t\in\sinkset}\binom{|\edgeset|}{r_t-\msgdim}q^{n_s-\msgdim+i-1},
  \end{equation}
  we have
  \begin{equation*}
    \ffield^{n_s} \setminus \cup_{t} (\Delta_t(\mathbf{0},r_t-\msgdim)+ 
      \langle \mathbf{g}_1,\cdots,\mathbf{g}_{i-1} \rangle)
      \neq \emptyset,
  \end{equation*}
  i.e., there exists a $g_i$ satisfying \eqref{eq:gti}.  Therefore, if
  $q$ satisfies \eqref{eq:fll} for all $i=1,\cdots, \msgdim$, or
  equivalently
  \begin{equation}\label{eq:upperboundqq} 
   q>   \sum_{t\in\sinkset}\binom{|\edgeset|}{r_t-\msgdim}, 
  \end{equation} 
   then there exists a vector that can be chosen as $\mathbf{g}_i$ for 
   $i=1,\cdots,\msgdim$. 
 
   Fix $\mathbf{g}_1,\cdots,\mathbf{g}_i$ that satisfy (\ref{eq:gti}). 
   We now prove by induction that (\ref{eq:gtg}) holds for $\mathbf{g}_1,\cdots,\mathbf{g}_i$. 
   If (\ref{eq:gtg}) does not hold for $i=1$, then there exists a 
   non-zero vector 
   $\alpha \mathbf{g}_1\in\Delta_t(\mathbf{0},r_t-\msgdim)$, where 
   $\alpha\in\mathbb{F}$. 
   Since $\Delta_t(\mathbf{0},r_t-\msgdim)$ is closed under scalar 
   multiplication and $\alpha\neq 0$, 
   we have $\mathbf{g}_1\in\Delta_t(\mathbf{0},r_t-\msgdim)$, a contradiction to 
   (\ref{eq:gti}) for $i=1$. 
   Assume  (\ref{eq:gtg}) holds for $i\leq k-1$. 
   If (\ref{eq:gtg}) does not hold for $i=k$, 
   then there exists a non-zero vector 
   \begin{equation*} 
     \sum_{i=1}^k\alpha_i 
   \mathbf{g}_i\in\Delta_t(\mathbf{0},r_t-\msgdim), 
   \end{equation*} 
   where $\alpha_i\in\mathbb{F} $. 
   If $\alpha_k=0$, 
   \begin{equation*} 
     \sum_{i=1}^{k-1}\alpha_i 
   \mathbf{g}_i\in\Delta_t(\mathbf{0},r_t-\msgdim), 
   \end{equation*} a 
   contradiction to the assumption that (\ref{eq:gtg}) holds for 
   $i=k-1$. Thus $\alpha_k\neq 0$. 
   Again, by $\Delta_t(\mathbf{0},r_t-\msgdim)$ being closed under 
   scalar multiplication, we have 
   \begin{align*} 
	\mathbf{g}_k & \in \Delta_t(\mathbf{0},r_t-\msgdim) 
        -\alpha_k^{-1}\sum_{i=1}^{k-1}\alpha_i \mathbf{g}_i\\ 
	 & \subset \Delta_t(\mathbf{0},r_t-\msgdim)+ 
         \langle \mathbf{g}_1,\cdots,\mathbf{g}_{k-1} \rangle, 
   \end{align*} 
   a contradiction to $\mathbf{g}_k$ satisfying (\ref{eq:gti}). 
   The proof is completed. 
\end{proof}

\subsection{The First Construction Algorithm} 
 
The proof of Theorem \ref{the:ach_sin} gives a construction
algorithm for an $(\msgdim, (r_t:t\in \sinkset),
(d_t:t\in\sinkset))$ linear network code and it also
verifies the correctness of the algorithm when the field
size is sufficiently large.  
This algorithm, called Algorithm 1, makes use of existing algorithms
(e.g., the Jaggi-Sanders algorithm \cite{poly}) to construct the local
encoding kernels.  The pseudo code of Algorithm 1 is shown below.
 
\begin{algorithm}[ht] \label{alg:alg-1} 
 \SetKwInOut{Output}{output} 
 \SetKwInOut{Input}{input} 
 \Input{($\graph$, $s$, $\sinkset$), ($r_t:t\in \sinkset$), $\msgdim$, ($d_t:t\in\sinkset$) with $r_t\leq \mc(s,t)\ \forall t\in\sinkset$ } 
 \Output{local encoding kernels and $\msgset$} 
 \SetLine 
 \Begin{ 
 \label{line:421} Construct a set of local encoding kernels such that $\rank(\bF_{s,t})=r_t$\;  
  \For{$i \leftarrow 1, \msgdim$}{ 
 \label{line:iai}  find $\mathbf{g}_i$ such that $\mathbf{g}_i\notin \cup_t\Delta_t(\mathbf{0},d_t-1) 
          + \langle \mathbf{g}_1,\cdots,\mathbf{g}_{i-1} \rangle$ \;  } 
  \caption{Construct network codes that achieve the refined Singleton bound.} \label{alg:2ka} 
  } 
\end{algorithm} 
 
The analysis of the complexity of the algorithm requires the following lemma implied by Lemma 5 and 8 in \cite{poly}. 
\begin{lemma}\label{lemma:1izea0} 
 Suppose $m\leq q$, the field size, and $\mathcal{B}_k\subset \ffield^n$, $k=1,\cdots, m$, are subspaces with $\dim(\mathcal{B}_k)<n$. A vector $\mathbf{u} \in \ffield^n\setminus \cup_{k=1}^{m} \mathcal{B}_k$ can be found in time $\bigO(n^3m+nm^2)$. 
\end{lemma} 
\begin{proof} 
 For each $\mathcal{B}_k$ find a vector $\mathbf{a}_k\in \ffield^n$ such that $\mathbf{a}_k\mathbf{b}^\tr =0$, $\forall \mathbf{b}\in \mathcal{B}_k$. This vector $\mathbf{a}_k$ can be obtained by solving the system of linear equations  
 \begin{equation*} 
  \mathbf{B}_k \mathbf{a}_k^\tr = \bzero, 
 \end{equation*} 
 where $\mathbf{B}_k$ is formed by juxtaposing a set of vectors that form a basis of $\mathcal{B}_k$. 
 The complexity of solving this system of linear equations is $\bigO(n^3)$.  
  
 We inductively construct $\mathbf{u}_1, \mathbf{u}_2, \cdots, \mathbf{u}_m$ such that $\mathbf{u}_i\mathbf{a}_k^\tr \neq 0$ for all $1\leq k\leq i\leq m$. If such a construction is feasible, then $\mathbf{u}_m \notin \mathcal{B}_k$, $\forall k\leq m$. Thus, $\mathbf{u}=\mathbf{u}_m\notin \cup_{k=1}^{m} \mathcal{B}_k$ is the desired vector. 
 
 Let $\mathbf{u}_1$ be any vector such that $\mathbf{u}_1\mathbf{a}_1^\tr \neq 0$.  For $1\leq i\leq m-1$, if $\mathbf{u}_i\mathbf{a}_{i+1}^\tr \neq 0$, we set $\mathbf{u}_{i+1}=\mathbf{u}_i$. Otherwise, find $\mathbf{b}_{i+1}$ such that $\mathbf{b}_{i+1}\mathbf{a}_{i+1}^\tr\neq 0$. We choose  
 \begin{equation} \label{eq:b9adf} 
 \alpha \in \ffield\setminus \{-(\mathbf{b}_{i+1}\mathbf{a}_{j}^\tr)/(\mathbf{u}_{i}\mathbf{a}_{j}^\tr): 1\leq j\leq i\}, 
 \end{equation} 
 and define  
 \begin{equation*} 
 \mathbf{u}_{i+1}=\alpha\mathbf{u}_i+\mathbf{b}_{i+1}. 
 \end{equation*} 
 The existence of such an $\alpha$ follows from $q\geq m >i$. 
 
 By construction, we know that 
 \begin{align*} 
  \mathbf{u}_{i+1}\mathbf{a}_{i+1}^\tr & = \alpha\mathbf{u}_i\mathbf{a}_{i+1}^\tr+\mathbf{b}_{i+1}\mathbf{a}_{i+1}^\tr \\ 
   & = \mathbf{b}_{i+1}\mathbf{a}_{i+1}^\tr \\ & \neq 0. 
 \end{align*} 
 If $\mathbf{u}_{i+1}\mathbf{a}_j^\tr =\alpha\mathbf{u}_i\mathbf{a}_{j}^\tr+\mathbf{b}_{i+1}\mathbf{a}_{j}^\tr =  0$ for some $1\leq j\leq i$, we have $\alpha = -(\mathbf{b}_{i+1}\mathbf{a}_{j}^\tr)/(\mathbf{u}_{i}\mathbf{a}_{j}^\tr)$, a contradiction to (\ref{eq:b9adf}). So, $\mathbf{u}_{i+1}\mathbf{a}_j^\tr\neq 0$ for all $j$ such that $1\leq j\leq i+1$.  
  
 Similar to the analysis in \cite[Lemma 8]{poly}, the construction of $\mathbf{u}$ takes time $\bigO(nm^2)$. Therefore, the overall time complexity is $\bigO(n^3m+nm^2)$. 
\end{proof} 
 
We analyze the time complexity of Algorithm 1 for the
representative special case that $r_t=r$ and $d_t=d$ for all
$t\in\sinkset$, where $r\leq \min_{t\in\sinkset}\mc(s,t)$
and $d\leq r-\msgdim+1$.  In the pseudo code, Line 2 can be
realized using the Jaggi-Sanders algorithm with complexity
$\bigO(|\edgeset||\sinkset|n(n+|\sinkset|))$, where
$n=\min_{t\in\sinkset}\mc(s,t)$ \cite{poly}.
Line 3-5 is a loop that runs Line 4 $\omega$
times. Considering $\Delta_t(\bzero, d-1)$ as the union of
$\binom{|\edgeset|}{d-1}$ subspaces  of $\ffield^r$, Line 4
can be realized in time
$\bigO(n_s^3|\sinkset|\binom{|\edgeset|}{d-1}+n_s(|\sinkset|\binom{|\edgeset|}{d-1})^2)$
as proved in Lemma \ref{lemma:1izea0}. Repeating $\msgdim$
times, the complexity of Line 3-5 is
\begin{equation*} 
  \bigO(\msgdim n_s^3|\sinkset|\xi+\msgdim n_s|\sinkset|^2\xi^2), 
\end{equation*} 
where $\xi = \binom{|\edgeset|}{d-1}$. 
The overall complexity is
\begin{equation*}
  \bigO(\msgdim n_s|\sinkset|\xi(n_s^2+|\sinkset|\xi) + |\edgeset||\sinkset|n(n+|\sinkset|))).
\end{equation*}
Comparing the complexities of constructing the local encoding kernels
(Line 2) and finding the codebook (Line 3-5), the latter term in the
above dominates when $d>1$.

To guarantee the existence of the code, we require the field size to
be sufficiently large.  From (\ref{eq:upperboundqq}) in the
proof of Theorem \ref{the:ach_sin}, all finite fields with
size larger than $|\sinkset|\binom{|\edgeset|}{r-\msgdim}$ are
sufficient. It is straightforward to show that this
algorithm can also be realized randomly with high success
probability if the field size is much larger than necessary.

\section{The Second Construction Algorithm}
\label{sec:algorithms}


Algorithm 1 can be regarded as finding a codebook for the given
transfer matrices. In this section, we study network error correction
from a different perspective by showing that we can also shape the
transfer matrices by designing proper local encoding
kernels. Following this idea, we give another algorithm that
constructs an $(\msgdim, (r_t:t\in \sinkset), (d_t:t\in\sinkset))$
linear network code.

\subsection{Outline of Algorithm 2}

We first give an informal description of this algorithm.
The second algorithm, called Algorithm 2, starts with a classical
error-correcting code as the codebook. The main task of the
algorithm is to design a set of local encoding kernels such
that the minimum distances of the network code, roughly speaking,
are the same as the classical error-correcting code.

It is complicated to design all the local encoding kernels
altogether.  Instead, we use an inductive method: we begin
with the simplest network that the source and the sink nodes
are directed connected with parallel edges; we then extend
the network by one edge in each iteration until the network
becomes the one we want.  For each iteration, we only need
to choose the local encoding kernels associated with the new
edge.

We have two major issues to solve in the above method: the first is
how to extend the network; the second is how to choose the local
encoding kernels. In Section \ref{sec:evol-netw-code}, we define a
sequence of networks $\graph^{i}$ for a given network $\graph$. The
first network is the simplest one as we described, the last one is the
network $\graph$, and $\graph^{i+1}$ has one more edge than
$\graph^i$.  In Section \ref{sec:algorithm}, we give an algorithm that
designs the local encoding kernels inductively. Initially, we choose a
classical error-correcting code that satisfies certain minimum
distance constraint. The local encoding kernels of $\graph^{i+1}$ is
determined as follows: Except for the new edge, all the local encoding
kernels in $\graph^{i+1}$ are inherited from $\graph^i$.  The new
local encoding kernels (associated with the new edge) is chosen to
guarantee 1) the preservation of the minimum distance of the network
code, and 2) the existence of the local encoding kernels to be chosen
in the next iteration.  We find a \emph{feasible condition} on the new
local encoding kernels to be chosen such that these criteria are
satisfied.

When $d_t=1$ for all sink nodes $t$, this algorithm degenerates to the
Jaggi-Sanders algorithm for designing linear network codes for the
error-free case.


\subsection{Iterative Formulation of Network Coding}  
\label{sec:evol-netw-code}  

In this and the next subsections, we describe the algorithm formally.
At the beginning, the algorithm finds $r_t$ edge-disjoint paths from
the source node $s$ to each sink node $\sinknode$ using a maximum flow
algorithm (for example, finding the augmenting paths).  We assume that
every edge in the network is on at least one of the
$\sum_{t\in\sinkset}r_t$ paths we have found. Otherwise, we delete the
edges and the nodes that are not on any such path, and consider the
coding problem for the new network.  Note that a network code for the
new network can be extended to the original network without changing
the minimum distances by assigning zero to all the local
encoding kernels associated with the deleted edges.

We consider a special order on the set of edges such that 1)
it is consistent with the partial order on the set of edges;
2) the first $n_s$ edges are in $\Out(s)$.
The order on the paths to a particular sink node is determined by the
first edges on the paths. 

Given a DAG $\graph$, we construct a sequence of graphs
$\graph^i=(\nodeset^i,\edgeset^i), i=0,1,\cdots, |\edgeset|-n_s$ as
follows. First, $\graph^0$ consists of a subgraph of $\mathcal{G}$
containing only the edges in $\Out(s)$ (and the associated nodes) and
all the sink nodes.  Following the order on $\mathcal{E}$, in the
$i$th iteration $\graph^{i-1}$ is expanded into $\graph^i$ by
appending the next edge (and the associated node) in $\mathcal{E}$.
This procedure is repeated until $\graph^i$ eventually becomes
$\mathcal{G}$. Note that $\graph^i$ contains $n_s+i$ edges and
$\graph^{|\edgeset|-n_s} =\graph$.  A sink node $t$ has $r_t$ incoming
edges in $\graph^i$, where the $j$th edge is the most downstream edge
in the truncation in $\graph^i$ of the $j$th edge-disjoint path from
the source node $s$ to sink node $t$ in $\mathcal{G}$.  With a slight
abuse of notation, we denote the set of incoming edges of a sink
node $t$ in $\graph^i$ as $\In(t)$, when $\graph^i$ is implied by the
context.  Fig. \ref{fig:8bafa} illustrates $\graph^0$ and $\graph^1$
when $\graph$ is the butterfly network.
 
 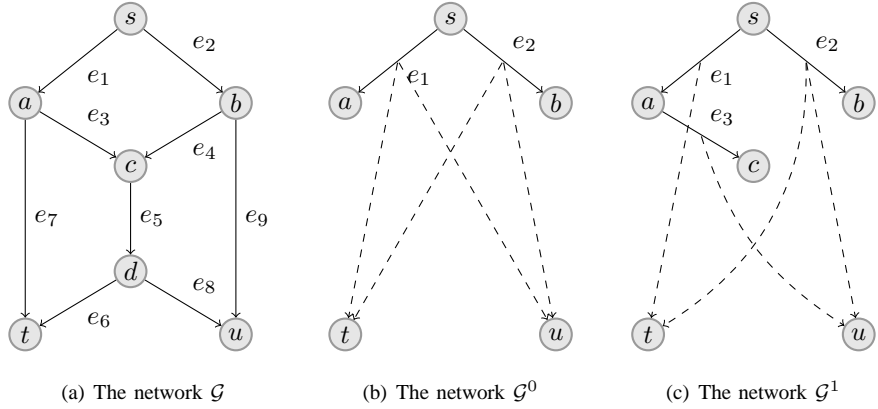
\begin{figure} 
	\centering  

    \subfigure[The network $\graph$]{ 
	\begin{tikzpicture}[scale=1.4] 
	\node[dot] (s) at (0,0) {$s$}; 
	\node[dot] (a) at (-1,-0.8) {$a$} edge [<-] node[auto,swap] {$e_1$} (s); 
	\node[dot] (b) at (1,-0.8) {$b$} edge [<-] node[auto,swap] {$e_2$} (s); 
	\node[dot] (c) at (0,-1.4) {$c$} edge [<-] node[auto,swap] {$e_4$} (b) edge [<-] node[auto,swap] {$e_3$} (a); 
	\node[dot] (d) at (0,-2.4) {$d$} edge [<-] node[auto,swap] {$e_5$} (c); 
	\node[dot] (t) at (-1,-3) {$t$} edge [<-] node[auto,swap] {$e_6$} (d) edge [<-] node[auto,swap] {$e_7$}  (a); 
	\node[dot] (u) at (1,-3) {$u$} edge [<-] node[auto,swap] {$e_8$}  (d) edge [<-] node[auto,swap] {$e_9$}  (b); 
	\end{tikzpicture} 
	} 
    \hspace{5pt} 
    \subfigure[The network $\graph^0$]{ 
        \begin{tikzpicture}[scale=1.4] 
	\node[dot] (s) at (0,0) {$s$}; 
	\node[dot] (a) at (-1,-0.8) {$a$} edge [<-] node[auto,swap] {$e_1$} (s); 
	\node[dot] (b) at (1,-0.8) {$b$} edge [<-] node[auto,swap] {$e_2$} (s); 
	\node[] (c) at (0,-1.4) {}; 
	\node[] (d) at (0,-2.4) {}; 
	\node[dot] (t) at (-1,-3) {$t$}; 
	\node[dot] (u) at (1,-3) {$u$}; 
	 
	\node[inner sep=0pt] (sa) at (-1/2,-0.8/2) {}; 
	\node[inner sep=0pt] (sb) at (1/2,-0.8/2) {}; 
	\draw [dashed, ->] (sa) to (t); 
	\draw [dashed, ->] (sa) to (u); 
	\draw [dashed, ->] (sb) to (t); 
	\draw [dashed, ->] (sb) to (u); 
	\end{tikzpicture} 
	} 
    \hspace{8pt} 
    \subfigure[The network $\graph^1$]{ 
	\begin{tikzpicture}[scale=1.4] 
	\node[dot] (s) at (0,0) {$s$}; 
	\node[dot] (a) at (-1,-0.8) {$a$} edge [<-] node[auto,swap] {$e_1$} (s); 
	\node[dot] (b) at (1,-0.8) {$b$} edge [<-] node[auto,swap] {$e_2$} (s); 
	\node[dot] (c) at (0,-1.4) {$c$} edge [<-] node[auto,swap] {$e_3$} (a); 
	\node[] (d) at (0,-2.4) {}; 
	\node[dot] (t) at (-1,-3) {$t$}; 
	\node[dot] (u) at (1,-3) {$u$}; 
	 
	\node[inner sep=0pt] (sa) at (-1/2,-0.8/2) {}; 
	\node[inner sep=0pt] (sb) at (1/2,-0.8/2) {}; 
	\node[inner sep=0pt] (ac) at (-1/2,-1.1) {}; 
	\draw [dashed, ->] (sa) to (t); 
	\draw [dashed, ->] (ac) to [bend right=20] (u); 
	\draw [dashed, ->] (sb) to [bend left=30] (t); 
	\draw [dashed, ->] (sb) to (u); 
	\end{tikzpicture} 
	} 
	\caption{An example of $\graph^0$ and $\graph^1$. The dashed lines are not new edges but indicate the incoming edges of $t$ and $u$. In $\graph^0$, both $t$ and $u$ have $e_1$ and $e_2$ as their incoming edges. In $\graph^1$, $\In(t)=\{e_1,e_2\}$ and $\In(u)=\{e_3,e_2\}$.}\label{fig:8bafa} 
 \end{figure}

The network $\graph^i$ is a multicast network with the source node $s$ and the set of sinks $\sinkset$. 
The algorithm chooses a proper codebook, and then constructs local encoding kernels starting with $\graph^0$. Except for the new edge, all the local encoding kernels in $\graph^{i+1}$ are inherited from $\graph^i$.  
We define ${\mathbf K}^i$, $\bF^i$, $\fout_{\rho}^i$, $\bz^i$ and $\bA_\rho^i$ for $\graph^i$ in view of ${\mathbf K}$, $\bF$, $\fout_{\rho}$, $\bz$ and $\bA_\rho$ defined for $\graph$  in Section~\ref{sec:prob-linear}, respectively. 
Writing $\fout_t^i= \fout_{\In(t)}^i$, 
we have  
\begin{equation} \label{eq:24kafa} 
\fout_t^i(\bx,\bz^i)=(\bx \bA^i_{\Out(s)}+\bz^i)  
\bF^i (\bA^i_{\In(t)})^\tr, 
\end{equation}  
in view of (\ref{eq:ak78a}). 
Further, we can define the minimum distance $d_{\min,t}^i$ corresponding to the sink node $t$ at the $i$th iteration as in~(\ref{eq:def8af}).

Consider a matrix $\mathbf M$. Let $(\mathbf M)_\idxset$ be the
submatrix of $\mathbf M$ containing the columns with indices in
$\idxset$, and $\mathbf M^{\backslash\idxset}$ be the submatrix
obtained by deleting the columns of $\mathbf M$ with indices in
$\idxset$.  If $\idxset=\{j\}$, we also write $\mathbf M^{\backslash
  j}$ and $(\mathbf{M})_j$ for $\mathbf M^{\backslash \{j\}}$ and
$(\mathbf{M})_{\{j\}}$, respectively.


In the following, we give an iterative
formulation of $\fout_t^i$ for $i>0$. 
Let $e$ be the edge added to $\graph^{i-1}$ to form $\graph^{i}$, and
let ${\bf k}_{e}=[\beta_{e',e}:e'\in \edgeset_{i-1}]$ be an
$(n_s+i-1)$-dimensional column vector.  In the $i$th
iteration, we need to determine the component $\beta_{e',e}$ of
$\mathbf{k}_e$ with $e'\in \In(\tail(e))$. All other components of
$\mathbf{k}_e$ are zero. 
Using $\mathbf{k}_e$, we have
\begin{align}  
  \bF^{i} & =  \left({\bf I}-{\bf K}^{i}\right)^{-1}
  \nonumber \\  
  & =  \left({\bf I}-\begin{bmatrix}{\bf K}^{i-1} &{\bf k}_{e} \\
      \bzero & 0\end{bmatrix}\right)^{-1} \nonumber \\ 
  & =  \begin{bmatrix}{\bf I}-{\bf K}^{i-1} &-{\bf k}_{e} \\
    \bzero & 1\end{bmatrix}^{-1}  \nonumber  \\ 
  & =  \begin{bmatrix}({\bf I}-{\bf K}^{i-1})^{-1} & ({\bf
      I}-{\bf K}^{i-1})^{-1}{\bf k}_{e} \\ \bzero & 1\end{bmatrix}  
  \nonumber \\ 
   & = \begin{bmatrix}\bF^{i-1} & \bF^{i-1}{\bf k}_{e}\\ \mathbf{0} &
     1\end{bmatrix}. \label{eq:852kg}  
\end{align}  
The matrix $\bA_{\Out(s)}^{i}$ has one more column with zero components than $\bA_{\Out(s)}^{i-1}$, i.e.,  
\begin{equation} \label{eq:8bak1b} 
  \bA_{\Out(s)}^{i}=  
  \begin{bmatrix}  
    \bA_{\Out(s)}^{i-1} & \mathbf{0}  
  \end{bmatrix}.  
\end{equation}  
If the edge $e$ is not on any path from the source node $s$  
to sink node $t$,  we only need to append a column with zero components to $\bA_{\In(t)}^{i-1}$  
to form $\bA_{\In(t)}^{i}$, i.e.,  
\begin{equation} \label{eq:biqlafk} 
  \bA_{\In(t)}^{i}=  
  \begin{bmatrix}  
    \bA_{\In(t)}^{i-1} & \mathbf{0}  
  \end{bmatrix}.  
\end{equation}  
For this case, we can readily obtain from (\ref{eq:24kafa}),  (\ref{eq:852kg}), (\ref{eq:8bak1b}) and (\ref{eq:biqlafk}) that 
\begin{align}  
  \fout_t^{i}(\bx,\bz^{i}) & = (\bx \bA_{\Out(s)}^{i}+\bz^{i}) \bF^{i} (\bA_{\In(t)}^i)^\tr \nonumber \\ 
    & = (\bx \bA_{\Out(s)}^{i}+\bz^{i}) \begin{bmatrix} \bF^{i-1} \\ \bzero \end{bmatrix} (\bA_{\In(t)}^{i-1})^\tr \nonumber \\ 
    & = (\bx \bA_{\Out(s)}^{i-1}+(\bz^{i})^{\backslash i})
    \bF^{i-1} (\bA_{\In(t)}^{i-1})^\tr \nonumber \\  
    & = \fout_t^{i-1}(\bx,(\bz^{i})^{\backslash i}). \label{eq:2028dq}  
\end{align}  
Note that $(\bz^{i})^{\backslash i}$ is an $(n_s+i-1)$-dimensional
error vector obtained by deleting the $i$th component of
$\bz^{i}$, which corresponds to $e$.

If edge $e$ is on the $j$th edge-disjoint path from the source node
$s$ to sink node $t$, to form $\bA_{\In(t)}^{i}$, we need to first
append a column with zero components to $\bA_{\In(t)}^{i-1}$, and then
move the `$1$' in the $j$th row to the last component of that row.
That is, if
\begin{equation*}
  \bA_{\In(t)}^{i-1} =\begin{bmatrix}  
  {\bf b}_1 \\ {\bf b}_2 \\ \vdots \\ {\bf b}_{r_t} \end{bmatrix},
\end{equation*}
then  
\begin{equation}  
  \label{eq:32q8ff}  
  \bA_{\In(t)}^{i} =  
  \begin{bmatrix}  
    {\bf b}_1  & 0 \\ \vdots & \vdots \\ {\bf b}_{j-1} & {0} \\  \bzero & 1 \\ {\bf b}_{j+1} & 0 \\ \vdots & \vdots \\ {\bf b}_{r_t} & 0  
  \end{bmatrix}.  
\end{equation}  

We can then obtain $F_t^{i}(\bx,\bz^{i})$ from (\ref{eq:24kafa}),  (\ref{eq:852kg}), (\ref{eq:8bak1b}) and (\ref{eq:32q8ff}) as 
\begin{align}  
  (\fout_t^{i}(\bx,\bz^{i}))_j  
  & =  
  (\bx\bA_{\Out(s)}^{i}+\bz^{i})\bF^{i}((\bA_{\In(t)}^{i})^\tr)_j 
  \label{eq:98eroa}    \\  
  & = (\bx\bA_{\Out(s)}^{i}+\bz^{i})\begin{bmatrix} \bF^{i-1}{\bf k}_{e} \\ 1 \end{bmatrix} \nonumber \\
  & =  (\bx\bA_{\Out(s)}^{i-1}+(\bz^{i})^{\backslash  
    i})\bF^{i-1}{\bf k}_{e} +(\bz^{i})_{i}, \nonumber
\end{align}  
and  
\begin{align}  
  \lefteqn{(\fout_t^{i}(\bx,\bz^{i}))^{\backslash {j}}} \nonumber \\ & =  
  (\bx\bA_{\Out(s)}^{i}+\bz^{i})\bF^{i}
  ((\bA_{\In(t)}^{i})^\tr)^{\backslash {j}} \nonumber \\ 
  & = (\bx\bA_{\Out(s)}^{i}+\bz^{i})\begin{bmatrix}\bF^{i-1} \\ \bzero \end{bmatrix}
  ((\bA_{\In(t)}^{i})^\tr)^{\backslash {j}}  \nonumber \\ 
  & =  (\fout_t^{i}(\bx,(\bz^{i})^{\backslash i}))^{\backslash {j}}.   \label{eq:13yrd}  
\end{align}  

\subsection{Algorithm 2}  
\label{sec:algorithm}  

\begin{algorithm}[htbp]\label{alg:18fa} 
  \label{alg:algorithm-2}  
   \SetKwInOut{Output}{output} 
 \SetKwInOut{Input}{input} 
 \Input{($\graph$, $s$, $\sinkset$), ($r_t:t\in \sinkset$), $\msgdim$, ($d_t$:$t\in\sinkset$)} 
 \Output{local encoding kernels and codebook $\mathcal C$}
   \SetLine 
   \Begin{ 
    \For{each sink node $\sinknode$}{  
      \label{l:b7a}choose $r_t$ edge disjoint paths from $s$ to $\sinknode$\;  
      initialize $\bA_{\In(t)}$\;  
    }  
    Find a linear codebook $\msgset$ with $d_{\min,t}^0\geq
    d_t$, $\forall t\in\sinkset$\;
    $\bF\leftarrow {\bf I}$, $\bA_{\Out(s)}\leftarrow {\bf I}$\; 
    \For{each $e\in\mathcal{E}\setminus \Out(s)$ in an  
       upstream to downstream order} { 
       $\Gamma\leftarrow \emptyset$\;  
       \For{each sink node $\sinknode$} { 
        \eIf{no chosen path from $s$ to $t$ crosses $e$} { 
          $\bA_{\In(t)} \leftarrow \begin{bmatrix} \bA_{\In(t)} &  
       \mathbf{0} \end{bmatrix}$\;} 
        ( \ $e$ is on the $j$th path from $s$ to $\sinknode$ ){ 
         \For{each $\idxset$ with $|\idxset|\leq d_t-1$ and  
       $j\notin \idxset$} { 
           \For{each $\rho$ with $|\rho|=d_t-1-|\idxset|$} { 
     	    \label{l:bi2akf}find $\bx_0\neq \mathbf{0}$ and $\bz_0$ matching  
     $\rho$ such that  
     $(\fout_t(\bx_0,-\bz_0))^{\setminus (\idxset\cup\{j\})}=  
     \mathbf{0}$\; \label{line:2ki} 
             \If{ exist $\bx_0$ and $\bz_0$ } { 
               \label{l:8bliq}$\Gamma \leftarrow \Gamma\cup\{{\bf k}$: $(\bx_0\bA-\bz_0)\bF{\bf k} =0 \}$\;\label{line:qii8} 
              } 
            }  
          } 
	} 
        update $\bA_{\In(t)}$ using (\ref{eq:32q8ff})\;  
      } 
    choose a vector ${\bf k}_e$ in $\mathbb{F}_q^{|\In(\tail(e))|}\setminus \Gamma$\; \label{line:28dao} 
    $\bF\leftarrow \begin{bmatrix}\bF & \bF {\bf k}_e\\ \mathbf{0} & 1\end{bmatrix}$\;  
   } 
   } 
    \caption{Construct $(\msgdim,
(r_t:t\in \sinkset), (d_t:t\in\sinkset))$ linear network
code}  
\end{algorithm}

Let $e$ be the edge appended to the graph in the $i$th iteration for
$i>0$.  We choose ${\bf k}_e$ such that the following \emph{feasible
  condition} is satisfied:
\begin{equation} 
(\fout_t^{i}(\bx,-\bz^{i}))^{\backslash\idxset}\neq \mathbf{0} 
  \label{eq:22rpo}  
\end{equation}  
for all combinations of 
\begin{enumerate} 
\renewcommand{\labelenumi}{C\arabic{enumi})} 
 \item \label{i:c1} $t\in\sinkset$, 
 \item \label{i:c2} $\idxset\subset \{1,2,\ldots,r_t\}$ with $0\leq |\idxset|\leq d_t-1$, 
 \item \label{i:c3} non-zero $\bx\in\mathcal{C}$, and  
 \item \label{i:c4} error vector $\bz^{i}$ with $w_H(\bz^{i})\leq d_t-1-|\idxset|$. 
\end{enumerate} 
If the feasible condition is satisfied for sink node $t$ and $\idxset
=\emptyset$, we have
\begin{equation*}
  \bx \bA_{\Out(s)}^{i}\bF^{i}(\bA_{\In(t)}^{i})^\tr\neq
  \bz^{i-1}\bF_t^{i},
\end{equation*}
for all $\bz^{i}$ and $\bx$ satisfying C\ref{i:c3} and C\ref{i:c4}.
If $\mathcal C$ is a subspace, we have
$d_{\min,t}^i \geq d_t$. Since the feasible condition is
required for each iteration, when the algorithm terminates,
the code constructed for $\graph$ satisfies $d_{\min,t}\geq
d_t$.  
Algorithm 2 is also called the \emph{distance preserving
  algorithm} since the algorithm keeps the minimum distance
larger than or equal to $d_{t}$ in each iteration.
Even though the feasible
condition is stronger than necessary for $d_{\min,t}^i\geq
d_t, t\in\sinkset$, as we will see, it is required for the
existence of the local encoding kernels for $k>i$ such that
the feasible condition is satisfied.

\begin{theorem} \label{the:0134kapqjf} Given a linear
  codebook with $d_{\min,t}^0 \geq d_t$ for all
  $t\in\sinkset$, there exist local encoding kernels such
  that the feasible condition is satisfied for $i= 1,\cdots,
  |\edgeset|-n_s$ when the field size is larger than
  $\sum_{t\in\sinkset}\binom{r_t+|\mathcal{E}|-2}{d_t-1}$.
\end{theorem} 
\begin{proof}[Proof Outline] 
  (See the complete proof in Section \ref{sec:verify}.) The linear
  codebook satisfies the feasible condition for $i=0$. Assume we can
  find local encoding kernels such that the feasible condition is
  satisfied for $i< k$, where $0\leq k-1<|\edgeset|-n_s$.  In the
  $k$th iteration, let $e$ be the edge appended to $\graph^{k-1}$ to
  form $\graph^{k}$.  We find that $\mathbf k_{e}$ only affects
  \eqref{eq:22rpo} for the case such that
  \begin{enumerate}
  \item $e$ is on $j$th path from $s$ to $t$, 
  \item $j\notin \idxset$, and
  \item $(F_t^{k-1}(\bx, -\bz))^{\backslash \idxset\cup
      \{j\}}=\bzero$, where $\bx\neq \bzero\in \mathcal{C}$, $\bz\in
    \ffield^{n_s+k-1}$, $w_H(\bz)=d_t-1-|\idxset|$.
  \end{enumerate}
  For $t$, $\idxset$, $\bx$ and $\bz$ satisfying the above
  condition, we need to choose $\mathbf k_e$ such that
  \begin{equation}\label{eq:constraint}
    (\bx\bA_{\Out(s)}^{k-1}-\bz)\bF^{k-1}\mathbf k_e\neq 0.
  \end{equation}
  We verify that if $q >
  \sum_{t\in\sinkset}\binom{r_t+|\mathcal{E}|-2}{d_t-1}$, we
  can always find such a $\mathbf k_e$.
\end{proof}

Refer to the pseudo code of Algorithm 2 above.
At the beginning, the algorithm finds $r_t$ edge-disjoint
paths from the source node to each sink node $t$, and
initializes $\bF$, $\bA_{\Out(s)}$, and $\bA_{\In(t)}, t\in \sinkset$ by
$\bF^0$, $\bA_{\Out(s)}^0$, and $\bA_{\In(t)}^0, t\in \sinkset$,
respectively.  The algorithm takes as the input a linear codebook
$\msgset$ such that $d_{\min,t}^0 \geq d_t$ for all sink
nodes $t$.  Such a codebook can be efficiently constructed by
using Reed-Solomon codes.
The main part of this algorithm is a loop starting at Line~7
for updating the local encoding kernels for the edges in
$\mathcal{E}\setminus \Out(s)$ in an upstream-to-downstream
order. The choosing of ${\bf k}_e$ is realized by the
pseudo codes between Line~8 and Line~25.

We analyze the time complexity of the algorithm for the
representative special case that $r_t=r$ and $d_t=d$ for all
$t\in\sinkset$, where $r\leq \min_{t\in\sinkset}\mc(s,t)$
and $d\leq r-\msgdim+1$.  For Line~3, the augmenting paths
for all the sinks can be found in time
$\bigO(|\sinkset||\edgeset|r)$ \cite{poly}.  Line~16 and 18
can be realized by solving a system of linear equations
which take time $\bigO(r^3)$ and $\bigO(1)$, respectively,
and each of these two lines is repeated
$\bigO(d|\edgeset||\sinkset|\binom{|\edgeset|}{d-1})$ times.
Line 26 can be solved by the method in Lemma
\ref{lemma:1izea0} in time
$\bigO(\delta|\sinkset|\binom{r+|\edgeset|-2}{d-1}(\delta^2+|\sinkset|\binom{r+|\edgeset|-2}{d-1}))$,
where $\delta$ is the maximum incoming degree of $\graph$,
and this line is repeated $\bigO(|\edgeset|)$ times. Under
the assumption that each edge is on some chosen path from
the source to the sinks, $\delta\leq r|\sinkset|$.  Summing
up all the parts, we obtain the complexity
\begin{align} 
 \bigO (\delta|\edgeset||\sinkset|\xi'(\delta^2 +|\sinkset|\xi')+ r^3d|\edgeset||\sinkset|\xi), 
\end{align} 
where $\xi'=\binom{r+|\edgeset|-2}{d-1}$. 
 
Subsequent to a conference paper of this work \cite{cons}, Matsumoto
\cite{Matsumoto2007} proposed an algorithm to construct network codes
that achieve the refined Singleton bound. In Table \ref{tab:13fag8a},
we compare the performances of Algorithm 1, Algorithm 2 and
Matsumoto's algorithm.  When $n_s$, $\omega$, $\delta$, $d$ and $r$
are fixed (i.e., we regard $|\sinkset|$ and $\edgeset$ as
variables) and $d>1$, the complexities of these algorithms are
$\bigO(|\sinkset|^2|\edgeset|^{2d-2})$,
$\bigO(|\sinkset|^2|\edgeset|^{2d-1})$ and
$\bigO(|\sinkset|^2|\edgeset|^{2d-1})$, respectively.

\begin{table}[htb]  
\renewcommand{\arraystretch}{1.3} 
\caption{Comparison of deterministic construction algorithms of network error-correcting codes.  $\xi=\binom{|\edgeset|}{d-1}$ and $\xi'=\binom{r+|\edgeset|-2}{d-1}$.} 
\centering \label{tab:13fag8a} 
 \begin{tabular}{l||ll} 
 \hline 
   & field size & Time complexity \\ 
   \hline \hline 
   Algorithm \ref{alg:alg-1} &  
   $|\sinkset|\xi$ & $\bigO(\msgdim n_s|\sinkset| \xi (n_s^2  +|\sinkset|\xi) + |\edgeset||\sinkset|n(n+|\sinkset|)))$ \\ 
   \hline 
   Algorithm \ref{alg:algorithm-2} & $|\sinkset|\xi'$ & $\bigO (\delta|\edgeset||\sinkset|\xi'(\delta^2 +|\sinkset|\xi')+ r^3d|\edgeset||\sinkset|\xi)$\\ 
   \hline 
   \cite[Fig. 2]{Matsumoto2007} & $|\sinkset|\xi$ & $\bigO(r|\edgeset||\sinkset|\xi(|\sinkset|\xi+r+d))$\\ 
   \hline 
 \end{tabular}  
\end{table}

\subsection{An Example of Algorithm 2}

	


We give an example of applying Algorithm 2 to
the network $(\graph, s, \{t,u\})$ shown in
Fig.~\ref{fig:complete}. In this network the maximum flow to each sink node
is three. We show how Algorithm 2 outputs a network code with $\omega = 1$,
$r_t=r_u=3$ and $d_{\min,t}=d_{\min,u}=3$. Here the finite field
$\ffield = \text{GF}(2^2) = \{0, 1, \alpha,\alpha^2\}$,
where $\alpha^2+\alpha + 1 = 0$.

The order on the set of edges is labelled in
Fig.~\ref{fig:complete}, and we also refer to an edge by its order. 
From $s$ to each sink node, there are three edge-disjoint
paths. We fix a particular path from $s$ to $t$ given by the sequence of edges
$3,6,8$ and a path from $s$ to $u$ given by the sequence of edges
$3,7,9$. The other edge-disjoint paths can be uniquely determined. 
We can check that each edge is on at least one path.
As we have described, define
$\graph^{0}=(\{s,a, d, e\},\{1,2,3\})$,
$\graph^{1}=(\{s,a, d, e\},\{1,2,3,4\})$ and so on.

We choose the codebook $\mathcal{C}=\langle
(1,\alpha,\alpha^2)\rangle$, which is a Reed-Solomon code. 
Let $\bx=(1,\alpha,\alpha^2)$. 
Note that we only need to check $\bx$ with the feasible
condition. The reason is that the constraint to choose
$\mathbf k_e$ in \eqref{eq:constraint} is
unchanged by multiplying a nonzero elements in $\ffield$ (see also
Section~\ref{sec:verify}).

Notice that nodes $b$, $c$, $d$ and $e$ have only one incoming edges.
We assume WLOG that the nodes $b$, $c$, $d$ and $e$ only copy and
forward their received symbols. We refer the reader to
\cite[Section~17.2]{yeung08b} for an explanation that this assumption does
not change the optimality of our coding design. 

In the following, we show that Algorithm 2 can give
$\beta_{3,6}=\beta_{4,6}=\beta_{3,7}=\beta_{5,7}=1$ and
$\beta_{5,6}=\beta_{4,7}=0$. Together with the local encoding kernels
associated with nodes $b$, $c$, $d$ and $e$, we have a set of local encoding
kernels satisfying the minimum distance constraints.

We skip the first two iterations, in which we assign $\beta_{1,4}=1$
and $\beta_{2,5}=1$. In the third iteration, edge $6$ is added to the
graph and we need to determine
\begin{equation*}
  \mathbf k_6 = \begin{bmatrix} 0 & 0 & \beta_{3,6} & \beta_{4,6}
    & \beta_{5,6} \end{bmatrix}^{\tr}.
\end{equation*}
We have
\begin{equation*}
  \bF^2 = \begin{bmatrix}1 & 0 & 0 & 1 & 0 \\ 0 & 1 & 0 & 0
    & 1 \\  0 & 0 & 1 & 0 & 0 \\ 0 & 0 & 0 & 1 & 0 \\ 0 &
    0 & 0 & 0 & 1 \end{bmatrix}.
\end{equation*}
We first consider node $t$. We see that edge $6$ is on the third
path to $t$.  In this iteration, $\In(t)=\{1,5,6\}$.  We
consider the following four cases of $\idxset$ such that
$3\notin \idxset$:
\begin{enumerate}
\item $\idxset=\emptyset$: Since $\bz_1 = (1,0,0, 0,\alpha)$
  satisfies $(F_t^2(\bx,-\bz_1))^{\backslash 3} = \bzero$, we
  need to choose $\mathbf k_6$ such that
  $(\bx\bA_{\Out(s)}^2-\bz_1)\bF^2\mathbf k_6 \neq 0$. This gives
  \begin{equation}
    \label{eq:con1}
    \beta_{3,6}\neq 0.
  \end{equation}
  We also have $\bz_2 = (1,\alpha,0, 0, 0)$
  satisfies $(F_t^2(\bx,-\bz_2))^{\backslash 3} = \bzero$. This
  error vector imposes the same constraint that
  $\beta_{3,6}\neq 0$.
\item $\idxset=\{1\}$: Since $\bz_3=(0,0,0,0,\alpha)$ satisfies
  $(F_t^2(\bx,-\bz_3))^{\backslash \{1,3\}} = 0$, we
  need to choose $\mathbf k_6$ such that
  $(\bx\bA_{\Out(s)}^2-\bz_3)\bF^2\mathbf k_6 \neq 0$. This gives
  \begin{equation}
    \label{eq:con2}
    \beta_{3,6}\alpha^2 + \beta_{4,6}\neq 0.
  \end{equation}
\item $\idxset=\{2\}$: Similar to the above case, we have
  \begin{equation}
    \label{eq:con3}
    \beta_{3,6}\alpha^2 + \beta_{5,6}\alpha\neq 0. 
  \end{equation}
\item $\idxset=\{1,2\}$: We need
  $\bx\bA_{\Out(s)}^2\bF^2\mathbf k_6 \neq 0$, i.e.,
  \begin{equation}
    \label{eq:con4}
    \beta_{3,6}\alpha^2+ \beta_{4,6} + \beta_{5,6}\alpha \neq 0. 
  \end{equation}
\end{enumerate}

Similarly, we can analyze sink node $u$ and 
obtain the following
constraints on $\mathbf k_6$:
\begin{align}
  \beta_{4,6} & \neq 0 \label{eq:con5} \\
  \beta_{3,6}\alpha^2 + \beta_{4,6} & \neq 0 \label{eq:con6}\\
  \beta_{4,6}+\beta_{5,6}\alpha & \neq 0  \label{eq:con7}\\
  \beta_{3,6}\alpha^2 + \beta_{4,6}+ \beta_{5,6}\alpha & \neq
  0.\label{eq:con8} 
\end{align}
Form \eqref{eq:con1} to \eqref{eq:con8}, we have six
distinct constraints,  which are satisfied by
$\beta_{3,6}=\beta_{4,6}=1$ and $\beta_{5,6}=0$.

\begin{figure} 
  \centering  
	
  \begin{tikzpicture}[scale=1.2] 
    \node[dot] (s) at (0,0) {$s$}; 
    \node[dot] (m1) at (-1.5,-1) {$d$} edge[<-] node[auto]
    {$1$} (s);
    \node[dot] (m2) at (1.5,-1) {$e$} edge[<-] node[right,
    above]
    {$2$} (s);
    \node[dot] (a) at (0,-1.6) {$a$} edge [<-] node[auto]
    {$3$} (s) edge [<-]  node[auto]
    {$4$} (m1) edge [<-] node[auto] {$5$} (m2);
    \node[dot] (b) at (-1,-2.5) {$b$} edge [<-] node[right]
    {$6$} (a);
    \node[dot] (c) at (1,-2.5) {$c$} edge [<-] node[left,right]
    {$7$} (a);
    \node[dot] (t) at (-2,-3.5) {$t$} edge [<-] node[left] {$10$} (m1) edge
    [<-] node[left] {$8$} (b) edge [<-] node[below] {$12$} (c);
    \node[dot] (u) at (2,-3.5) {$u$} edge [<-] node[right] {$11$} (m2) edge [<-]
    node[below] {$13$} (b) edge [<-] node[right] {$9$} (c);
  \end{tikzpicture}

  \caption{This network is used to demonstrate Algorithm 2, in which
    $s$ is the source node, $t$ and $u$ are the sink nodes. The edges
    in the network is labelled by the integers beside.  We design
    a code with $\omega = 1$, $r_t=r_u=3$ and
    $d_{\min,t}=d_{\min,u}=3$ over $\text{GF}(2^2)$.}\label{fig:complete}
\end{figure}
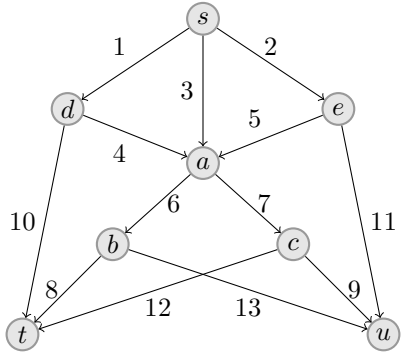

Then we go to the fourth iteration, for which edge $7$ is added to the
graph and we need to determine
\begin{equation*}
  \mathbf k_7 = \begin{bmatrix} 0 & 0 & \beta_{3,7} & \beta_{4,7}
    & \beta_{5,7} & 0\end{bmatrix}^{\tr}.
\end{equation*}
We have
\begin{equation*}
  \bF^3 = \begin{bmatrix}1 & 0 & 0 & 1 & 0 & 1 \\ 0 & 1 & 0 & 0
    & 1 & 0 \\  0 & 0 & 1 & 0 & 0 & 1 \\ 0 & 0 & 0 & 1 & 0 &
    1 \\ 0 & 0 & 0 & 0 & 1 & 0 \\ 0 & 0 & 0 & 0 & 0 & 1 \end{bmatrix}.
\end{equation*}
Edge $7$ is on the second path to $t$. Considering all $\idxset$ such
that $2\notin \idxset$, we obtain the following constraints on $\mathbf k_7$:
\begin{align}
  \beta_{5,7} & \neq 0 \label{eq:don1} \\
  \beta_{3,7}\alpha^2 + \beta_{5,7}\alpha & \neq 0 \\
  \beta_{3,7} + \beta_{4,7}+ \beta_{5,7}\alpha & \neq 0 \\
  \beta_{3,7}\alpha^2 + \beta_{4,7}\alpha^2+ \beta_{5,7}\alpha & \neq 0 \\
  \beta_{3,7}\alpha^2 + \beta_{5,7}\alpha & \neq 0 \\
  \beta_{3,7}\alpha^2+ \beta_{4,7} + \beta_{5,7}\alpha & \neq 0.
\end{align}
Similarly, we can analyze sink node $u$ and obtain the following
constraints on $\mathbf k_7$:
\begin{align}
  \beta_{3,7}+ \beta_{4,7} & \neq 0 \label{eq:don6} \\
  \beta_{3,7}\alpha^2+ \beta_{4,7} & \neq 0 \label{eq:don7} \\
  \beta_{3,7}\alpha^2+ \beta_{4,7}\alpha^2+\beta_{5,7}\alpha
  & \neq 0  \label{eq:don8}\\ 
  \beta_{3,7}\alpha^2 + \beta_{4,7}+\beta_{5,7}\alpha & \neq
  0\label{eq:don9} 
\end{align}
From \eqref{eq:don1} to \eqref{eq:don9}, we have seven
distinct constraints on $\mathbf k_7$, which are satisfied
by $\beta_{3,7}=\beta_{5,7}=1$ and $\beta_{4,7}=0$.

Let us see what would happen if we only consider
$\idxset=\emptyset$.  For this case, in iteration $3$, we have only
two constraints given by \eqref{eq:con1} and \eqref{eq:con5}, which
are satisfied by $\beta_{3,6}=\beta_{4,6}=1$ and
$\beta_{5,6}=\alpha$. We see that these values do not satisfiy
\eqref{eq:con3}.  We now show that it is impossible to find a network
code with $d_{\min,t}=3$ with these values. Construct an error vector
$\bz$ as follows:  $z_1=1$, 
$z_7=-(\beta_{3,7}\alpha^2+\beta_{5,7}\alpha)$ and $z_i=0$ for $i\neq 1,
7$.  We check that
\begin{align*}
  F_t(\bx,\bz) & = \bx \bF_{s,t} + \bz \bF_t \\
  & = (1,\alpha,\alpha^2) \begin{bmatrix}1 & \beta_{4,7} &
    1 \\ 0 & \beta_{5,7} & \alpha & 0 & \beta_{3,7} &
    1  \end{bmatrix} + (1, z_7) \begin{bmatrix} 1 &
    \beta_{4,7} & 1 \\ 0 & 1 & 0 \end{bmatrix} \\
  & = \bzero.
\end{align*}
Thus, $d_{\min,t}\leq w_H(\bz)=2$.


\subsection{Proof of Theorem~\ref{the:0134kapqjf}}
\label{sec:verify}  
 
Theorem \ref{the:0134kapqjf} is proved by induction on $i$.  
The codebook with $d_{\min,t}^0 \geq d_t$ for all $t\in
\sinkset$ satisfies the feasible condition for $i=0$. 
Assume that up to the $(k-1)$th iteration, where $0\leq
k-1<|\edgeset|-n_s$, we can find local encoding kernels such
that the feasible condition is satisfied for all $i \le k$.
In the $k$th iteration, let $e$ be the edge appended to
$\graph^{k-1}$ to form $\graph^{k}$.  We will show that there
exists ${\bf k}_e$ such that the feasible condition
continues to hold for $i = k$.

We first consider a sink node $t$ for which edge $e$ is not on any path from  
the source node $s$ to $\sinknode$. (Such a sink node does
not necessarily exist).   
For all $\idxset$, $\bx$ and $\bz^{k}$ satisfying
C\ref{i:c2})-C\ref{i:c4}) with $k$ in place of $i$,  
we have 
\begin{align} 
    (\fout_\sinkt^{k}(\bx,-\bz^{k}))^{\backslash\idxset}
    & = (\fout_\sinkt^{k-1}(\bx,-(\bz^{k})^{\backslash
      k}))^{\backslash\idxset} \label{eq:18akee} \\ & \neq
    \mathbf{0}, \label{eq:8afla}  
\end{align} 
where (\ref{eq:18akee}) follows from (\ref{eq:2028dq}), and
(\ref{eq:8afla}) follows from the induction hypothesis, i.e.,
the feasible condition is satisfied for $i=k-1$ by noting
$w_H((\bz^{k})^{\setminus k})\leq w_H(\bz^{k})\leq d_t-1-|\idxset|$. 
Therefore, (\ref{eq:22rpo}) holds for $i=k$ regardless of
the choice of ${\bf k}_e$.

For a sink node $t$ such that edge $e$ is on the $j$th edge-disjoint  
path from the source node $s$ to $\sinknode$, we consider  
two scenarios for $\mathcal{L}$, namely  
$j\in \idxset$ and $j\notin \idxset$.  
For all $\mathcal{L}$ satisfying C\ref{i:c2}) and $j\in
\idxset$, and all $\bx$ and $\bz^{k}$ satisfying
C\ref{i:c3}) and C\ref{i:c4}) for $i=k$,  
\begin{align} 
    (\fout_\sinkt^{k}(\bx,-\bz^{k}))^{\backslash\idxset}  & = (\fout_\sinkt^{k-1}(\bx,-(\bz^{k})^{\backslash k}))^{\backslash\idxset} \label{eq:1irqfi}\\ & \neq \mathbf{0}, \label{eq:8g1oapq} 
\end{align} 
where (\ref{eq:1irqfi}) follows from (\ref{eq:13yrd}) and
(\ref{eq:8g1oapq}) follows from the induction hypothesis
using the same argument as the previous case. Therefore,
(\ref{eq:22rpo}) again holds for $i=k$ regardless of the
choice of ${\bf k}_e$.

For all $\mathcal{L}$ satisfying C\ref{i:c2}) and $j \not\in \idxset$, all $\bx$ satisfying C\ref{i:c3}) and all $\bz^{k}$ satisfying C\ref{i:c4}) with $i=k$, 
(\ref{eq:22rpo}) holds for $i=k$ if and only if either 
\begin{equation}  
(\fout_t^{k}(\bx,-\bz^{k}))^{\backslash\idxset\cup\{j\}} \ne \bzero  
\label{B1}  
\end{equation}  
or  
\begin{equation}  
(\fout_t^{k}(\bx,-\bz^{k}))_j \ne 0. 
\label{B2}
\end{equation}  
By (\ref{eq:13yrd}) and (\ref{eq:98eroa}), (\ref{B1}) and (\ref{B2})  
are equivalent to  
\begin{equation}\label{eq:25qea}  
  (\fout_t^{k-1}(\bx,-(\bz^{k})^{\backslash k}))^{\backslash\idxset\cup\{j\}}  
\neq \mathbf{0},  
\end{equation}  
and  
\begin{equation}\label{eq:26d2a}  
  (\bx\bA_{\Out(s)}^{k-1}-(\bz^{k})^{\backslash k}))\bF^{k-1}{\bf k}_e  
  -(\bz^{k})_{k}
  \neq 0,  
\end{equation}  
respectively. Note that ${\bf k}_e$ is involved in (\ref{eq:26d2a}) but not in (\ref{eq:25qea}).

For an index set $\idxset$ satisfying C\ref{i:c2}) and $j \not\in \idxset$, let $\badxz^{k}_\idxset$ be the set of all $(\bx,\bz^{k})$ that do  
not satisfy (\ref{eq:25qea}), where $\bx$ satisfies C\ref{i:c3}) and $\bz^{k}$ satisfies C\ref{i:c4}) for $i=k$.  
We need to find a proper ${\bf k}_e$ such that for any $(\bx,\bz^{k})\in \badxz^{k}_\idxset$,  $(\bx,\bz^{k})$ satisfies (\ref{eq:26d2a}). 
In the following technical lemmas, we first prove some properties of $\badxz^{k}_\idxset$. 
 
\begin{lemma}\label{lemma:18saa}  
If the feasible condition holds for $i=k-1$, then for any $(\bx,\bz^{k})\in\badxz^{k}_\idxset$,  
$w_H(\bz^{k})=d_t-1-|\idxset|$ and $(\bz^{k})_{k}=0$.  
\end{lemma}  
\begin{proof}  
Fix $(\bx,\bz^{k})\in\badxz^{k}_\idxset$. 
  If $|\idxset|=d_t-1$, since  
$w_H(\bz^{k})\leq d_t-1-|\idxset|=0$,  
the lemma is true.  
If $0\leq |\idxset|< d_t-1$, we now prove that  
$w_H((\bz^{k})^{\backslash k})> d_t-2-|\idxset|$. 
If $w_H((\bz^{k})^{\backslash k}) \leq d_t-2-|\idxset|$,  
by the assumption that the feasible condition holds for $i=k-1$,  
\begin{align} 
 (F_t^{k-1}(\bx,-(\bz^{k})^{\backslash k}))^{\idxset\cup \{j\}} \neq \bzero, 
\end{align} 
i.e., $(\bx,\bz^{k})$ satisfies (\ref{eq:25qea}), a contradiction to $(\bx,\bz^{k})\in\badxz^{k}_\idxset$. 
Therefore 
\begin{align} 
 d_t-1-|\idxset| & \leq w_H((\bz^{k})^{\backslash k}) \\ 
  & \leq  w_H(\bz^{k}) \\ 
  & \leq  d_t-1-|\idxset|. 
\end{align} 
Hence, $w_H((\bz^{k})^{\backslash k}) = w_H(\bz^{k}) =  d_t-1-|\idxset|$. This also implies that $(\bz^{k})_{k}=0$. 
\end{proof}

\begin{lemma} \label{lemma:4}  
Let $\mathbf M$ be a matrix, and let $j$
be a column index of $\mathbf M$.  
If a system of linear equations $x \mathbf M=\bzero$ with
$x$ as the variable has only the zero solution, then  
$x \mathbf M^{\backslash j}=\bzero$ has at most a one-dimensional  
solution space. 
\end{lemma}  
\begin{proof}  
  The number of columns of $M$ is at least the number of rows of $M$,
  otherwise the system of linear equations $x \mathbf M=\bzero$ cannot
  have a unique solution.  Let $m$ be the number of rows in $\mathbf
  M$. We have $\rank(\mathbf M)=m$.  Let $\text{Null}(\mathbf
  M^{\backslash j})$ be the null space of $\mathbf M^{\backslash j}$
  defined as
\begin{equation*}
  \text{Null}(\mathbf M^{\backslash j}) = \{\bx: \bx \mathbf M^{\backslash j} = 0\}.
\end{equation*}
By the rank-nullity theorem of linear algebra, we have
$$\rank(\mathbf M^{\backslash j}) + \dim(\text{Null}(\mathbf
M^{\backslash j})) = m.$$  
Hence,
\begin{align*}
  \dim(\text{Null}(\mathbf M^{\backslash j})) 
   & = \rank(\mathbf M)-\rank(\mathbf M^{\backslash j})  \\
   & \leq 1.
\end{align*}
The proof is completed by noting that $\text{Null}(\mathbf
M^{\backslash j})$ is the solution space of $x \mathbf
M^{\backslash j}=\bzero$ with $x$ as the variable.
\end{proof}

\begin{lemma} \label{lemma:1d8a}  
Let $\rho$ be an error pattern  with $|\rho|=d_t-1-|\idxset|$, where $0\leq |\idxset|\leq d_t-1$. If the feasible condition holds for $i=k-1$, the span of  
all $(\bx,\bz^{k})\in\badxz^{k}_\idxset$ with $\bz^{k} \in \rho^*$ is either empty or a one-dimensional linear space.   
\end{lemma}  
\begin{proof}  
Consider the equation  
\begin{equation} \label{eq:8zlida} 
 (\fout_t^{k-1}(\bx,-(\bz^{k})^{\backslash k}))^{\backslash\idxset} = \mathbf{0} 
\end{equation} 
with $\bx \in \msgset$ and $\bz^{k} \in \rho^*$ as variables.  
Since $\msgset$ and $\rho^*$ are both vector spaces, (\ref{eq:8zlida}) is a system of linear equations. 
By the assumption that the feasible condition holds for $i=k-1$,  
(\ref{eq:8zlida}) has only the zero solution.  
By Lemma \ref{lemma:4}, the system of linear equations 
\begin{equation*}  
 (\fout_t^{k-1}(\bx,-(\bz^{k})^{\backslash k}))^{\backslash\idxset\cup\{j\}} = \mathbf{0}, 
\end{equation*} 
with $\bx \in \msgset$ and $\bz^{k} \in \rho^*$ as variables, has at most a one-dimensional solution space.  
\end{proof}  
 
 
\begin{lemma}\label{lemma:5}  
If the feasible condition holds for $i=k-1$, there exist at most  
$\binom{n_s+k-1}{d_t-1-|\idxset|}q^{|\In(\tail(e))|-1}$ values of 
${\bf k}_e$ such that (\ref{eq:26d2a}) does not hold  for some  
$(\bx,\bz^{k})\in\badxz^{k}_\idxset$.  
\end{lemma}  
\begin{proof}  
For $(\bx_0,\bz_0^{k})\in \badxz^{k}_\idxset$, by Lemma~\ref{lemma:18saa}, $(\bz_0^{k})_{k}=0$. 
Thus, all the ${\bf k}_e$ satisfying  
\begin{equation}  
  \label{eq:28jjr4}  
  (\bx_0\bA_{\Out(s)}^k-(\bz_0^{k})^{\backslash k})\bF^k{\bf k}_e  =  0  
\end{equation}  
do not satisfy (\ref{eq:26d2a}) for $(\bx_0,\bz^{k}_0)\in
\badxz^{k}_\idxset$.    
To count the number of solutions of (\ref{eq:28jjr4}), we notice that 
\begin{align} 
 (\fout_t^{k-1}(\bx_0,-(\bz_0^{k})^{\backslash k}))^{\backslash\idxset} \neq \mathbf{0}, 
\end{align} 
by the feasible condition holding for $i=k-1$,  
and 
\begin{align} 
 (\fout_t^{k-1}(\bx_0,-(\bz_0^{k})^{\backslash k}))^{\backslash\idxset\cup\{j\}} = \mathbf{0}, 
\end{align} 
since $(\bx_0,\bz^{k}_0)\in \badxz^{k}_\idxset$. 
Thus,  
\begin{equation*}
 (\fout_t^{k-1}(\bx_0,-(\bz_0^{k})^{\backslash k}))_j= ((\bx_0\bA_{\Out(s)}^{k-1}-(\bz_0^{k})^{\backslash k})\bF^{k-1}(\bA_{\In(t)}^{k-1})^\tr)_j\neq \mathbf{0},
\end{equation*} 
which gives a nonzero component of
$(\bx_0\bA_{\Out(s)}^{k-1}-(\bz_0^{k})^{\backslash k})\bF^{k-1}$ corresponds
to the edge that precedes edge $e$ on the $j$th path from $s$ to $t$.
This shows that the components of
$(\bx_0\bA_{\Out(s)}^{k-1}-(\bz_0^{k})^{\backslash k})\bF^{k-1}$ corresponding to the edges in $\In(\tail(e))$ are
not all zero. On the other hand, a component of ${\bf k}_e$ can
possibly be nonzero if and only if it corresponds to an edge in
$\In(\tail(e))$.  Therefore, the solution space of ${\bf k}_e$ in
(\ref{eq:28jjr4}) is an $\mathbb{F}_q^{|\In(\tail(e))|-1}$-dimensional
subspace.

By Lemma \ref{lemma:18saa}, for each $(\bx,\bz^{k})\in
\badxz^{k}_\idxset$, $\bz^{k}$ must match an error
pattern $\rho$ with $|\rho|=d_t-1-|\idxset|$ and $e\notin
\rho$.  
Since there are totally $n_s+k-1$ edges in $\graph^{k}$
excluding $e$, there are $\binom{n_s+k-1}{d_t-1-|\idxset|}$
error patterns with size $d_t-1-|\idxset|$.

Consider an error pattern $\rho$ with $|\rho|=d_t-1-|\idxset|$ and $e\notin \rho$. 
By Lemma \ref{lemma:1d8a}, if $(\bx_0,\bz^{k}_0)\in \badxz^{k}_\idxset$ with $\bz^{k}_0\in\rho^*$, all $(\bx,\bz^{k})\in \badxz^{k}_\idxset$ with $\bz^{k}\in \rho^*$ can be expressed as $(\alpha\bx_0,\alpha\bz^{k}_0)$ with nonzero $\alpha\in \ffield$. 
Since we obtain the same solutions of ${\bf k}_e$ in (\ref{eq:28jjr4}) when $\bx_0$ and $\bz_0^{k}$ are replaced by $\alpha\bx_0$ and $\alpha\bz^{k}_0$, respectively,  
for a particular pattern $\rho$, we only need to consider one $(\bx_0,\bz^{k}_0)\in \badxz^{k}_\idxset$ with $\bz^{k}_0\in\rho^*$. 
 
Upon considering all error patterns $\rho$ with $|\rho|=d_t-1-|\idxset|$ and $e\notin \rho$, we conclude that there exist at most  
$\binom{n_s+k-1}{d_t-1-|\idxset|}q^{|\In(\tail(e))|-1}$ values of ${\bf k}_e$ not  
satisfying $(\ref{eq:26d2a})$ for some $(\bx,\bz^{k})\in\badxz^{k}_\idxset$. 
\end{proof}  
  
Considering the worst case that for all $\sinkt\in\sinkset$, edge $e$  
is on an edge-disjoint path  
from the source node $s$ to sink node $\sinkt$, and considering all the index set $\idxset$ with $0\leq |\idxset|\leq d_t-1$ and $j\notin\idxset$ for each sink node $t$, we have at most  
\begin{eqnarray}  
& & \sum_{t\in \sinkset}\sum_{l=0}^{d_t-1}  
\binom{r_t-1}{l}\binom{n_s+k-1}{d_t-1-l}q^{|\In(\tail(e))|-1} \nonumber\\  
&= &\sum_{t\in \sinkset}\binom{r_t+n_s+k-2}{d_t-1} q^{|\In(\tail(e))|-1} \\ 
& \leq & \sum_{t\in\mathcal{T}}\binom{r_t+|\mathcal{E}|-2}{d_t-1}q^{|\In(\tail(e))|-1} \label{eq:32kaoq} 
\end{eqnarray}  
vectors that cannot be chosen as ${\bf k}_e$. Note that (\ref{eq:32kaoq}) is justified because $0\leq k-1< |\edgeset|-n_s$. 
Since $q>\sum_{t\in\mathcal{T}}\binom{r_t+|\mathcal{E}|-2}{d_t-1}$,  
there exists a choice of ${\bf k}_e$ such that  
for all $\mathcal{L}$ satisfying C\ref{i:c2}) and $j \not\in \idxset$, all $\bx$ satisfying C\ref{i:c3}), and all $\bz^{k}$ satisfying C\ref{i:c4}) for $i=k$, 
(\ref{eq:22rpo}) holds for $i=k$. Together with the other cases (where the choice of ${\bf k}_e$ is immaterial), we have proved the existence of a ${\bf k}_e$ such that the feasible condition holds for $i=k$. 
 

\section{Concluding Remarks}

This work, together with the previous work \cite{yang08eu},
gives a framework for coherent network error correction. The
work \cite{yang08eu} characterizes the error
correction/detection capability of a general transmission
system with network coding being a special case.  The
problems concerned here are the coding bounds and the code
construction for network error correction.

In this work, refined versions of the Hamming
bound, the Singleton bound and the Gilbert-Varshamov bound
for network error correction have been presented with simple
proofs based on the distance measures developed in
\cite{yang08eu}. These bounds are improvements over the ones
in \cite{nwc_err, nec1,nec2} for the linear network coding
case.  Even though these bounds are stated based on the
Hamming weight as the weight measure on the error vectors,
they can also be applied to the weight measures in
\cite{zhang08, adv, weight} because of the equivalence
relation among all these weight measures (See
\cite{yang08eu,yangthesis}).

Like the original version of the Singleton bound
\cite{nwc_err, nec1}, the refined Singleton bound for linear
network codes proved in this paper continues to be
tight. Two different construction algorithms have been
presented and both of them can achieve the refined Singleton
bound.  The first algorithm finds a codebook based on a
given set of local encoding kernels, which simply constructs
an MDS code when the problem setting is the classical case.
The second algorithm constructs a set of of local encoding
kernels based on a given classical error-correcting code
satisfying a certain minimum distance requirement by
recursively choosing the local encoding kernels that
preserve the required minimum distance properties.


There are many problems to be solved towards application of network
error correction. Our algorithms require a large field size to
guarantee the existence of network codes with large minimum
distances. One future work is to consider how to relax this field size
requirement. Fast decoding algorithms of network error-correcting
codes are also desired.  Moreover, network error correction in cyclic
networks is sill lack of investigation.

\section*{Acknowledgment}

The authors thank the reviewers and the associate
editor for their precise and insightful comments. 
Those comments help us to improve various aspects of this paper.
The authors also thank Prof. Zhen Zhang for
the valuable discussion and Prof. Sidharth Jaggi for his suggestions.

The work of Raymond W.~Yeung was partially supported
by a grant from the Research Grant Committee of the Hong Kong 
Special Administrative Region, China (RGC Ref.~No.\
CUHK 2/06C) and a grant from Cisco System, Inc.




\end{document}